\documentclass[10pt]{article}
\usepackage{amsthm,amsfonts,amsmath,amssymb}
\usepackage{geometry}
\usepackage{bbm}		 			
\usepackage[linesnumbered,ruled]{algorithm2e} 	
\usepackage{float}					
\usepackage{hyperref}				
\usepackage{bm}						
\usepackage{enumitem}				
\usepackage{color}
\usepackage{adjustbox}
\allowdisplaybreaks					
\newtheorem{theorem}{Theorem}[section]

\newtheorem{lemma}[theorem]{Lemma}
\newtheorem{informal theorem}[theorem]{Theorem (informal statement)}

\newtheorem{corollary}[theorem]{Corollary}

\theoremstyle{definition}
\newtheorem{definition}[theorem]{Definition}

\newcommand{\alg}{\text{ALG}}

\newcommand{\opt}{\text{OPT}}
\newcommand{\al}{\mathcal{A}}
\newcommand{\boxes}{\mathcal{B}}
\newcommand{\scenario}{\mathcal{S}}
\newcommand{\step}{\mathcal{T}}
\newcommand{\e}{\varepsilon}      			
\newcommand{\lp}{\left}						
\newcommand{\rp}{\right}
\newcommand{\E}[2]{\textbf{E}_{#1}\lp[ #2 \rp]}	
\renewcommand{\Pr}[1]{\textbf{Pr}\lp[ #1 \rp]}	
\newcommand{\cost}{\operatorname{cost}}
\newcommand{\poly}{\operatorname{poly}}
\newcommand{\dist}{\mathcal{D}}
\newcommand{\feas}{\mathcal{F}}
\newcommand{\timeset}{\mathcal{T}}

\newcommand{\kconst}{\alpha}
\newcommand{\mconst}{\alpha}

\newcommand{\remove}[1]{}

\title{Pandora's Box with Correlations:\\ Learning and Approximation}
\author{
Shuchi Chawla \\ UW-Madison \\ {\tt shuchi@cs.wisc.edu} \and 
Evangelia Gergatsouli \\ UW-Madison \\ {\tt gergatsouli@wisc.edu} \and 
Yifeng Teng \\ UW-Madison \\ {\tt yifengt@cs.wisc.edu} \and 
Christos Tzamos \\ UW-Madison \\ {\tt tzamos@wisc.edu} \and 
Ruimin Zhang \\ UW-Madison \\ {\tt rzhang274@wisc.edu} 
}

\begin{document}

\maketitle

\begin{abstract}
The Pandora's Box problem and its extensions capture optimization problems with
stochastic input where the algorithm can obtain instantiations of input random
variables at some cost. To our knowledge, all previous work on this class of
problems assumes that different random variables in the input are distributed
independently.  As such it does not capture many real-world settings. In this
paper, we provide the first approximation algorithms for Pandora's Box-type
problems with correlations. We assume that the algorithm has access to samples
drawn from the joint distribution on input.

Algorithms for these problems must determine an order in which to probe random
variables, as well as when to stop and return the best solution found so far.
In general, an optimal algorithm may make both decisions adaptively based on
instantiations observed previously. Such {\em fully adaptive} (FA) strategies
cannot be efficiently approximated to within any sub-linear factor with sample
access. We therefore focus on the simpler objective of approximating {\em
partially adaptive} (PA) strategies that probe random variables in a fixed
predetermined order but decide when to stop based on the instantiations
observed. We consider a number of different feasibility constraints and provide
simple PA strategies that are approximately optimal with respect to the best PA
strategy for each case. All of our algorithms have polynomial sample
complexity. We further show that our results are tight within constant factors:
better factors cannot be achieved even using the full power of FA strategies. 
\end{abstract}

\setcounter{page}{0}
\thispagestyle{empty}
\newpage

\section{Introduction}\label{sec:intro}

In many optimization settings involving uncertainty in the input, information
about the input can be obtained at extra monetary or computational overhead;
paying this overhead can allow the optimizer to improve its performance.
Determining the optimal manner for acquiring information then becomes an online
decision-making problem: each piece of information obtained by the algorithm
can affect whether and which piece to acquire next. A classical example is the
Pandora's Box problem due to Weitzman \cite{Weit1979}. The online algorithm is
presented with $n$ boxes, each containing an unknown stochastic reward. The
algorithm can open boxes in any order at a fixed overhead each; observes the
rewards contained in the open boxes; and terminates upon selecting any one of
the rewards observed. The goal is to maximize the reward selected minus the
total overhead of opening boxes. Weitzman showed that a particularly simple
policy is optimal for the Pandora's Box problem: the algorithm computes an
index for each box based on its reward distribution and opens boxes in
decreasing order of these indices until it finds a reward that exceeds all of
the remaining indices.  There is a long literature of generalizations of this
problem, in many different settings \cite{ CharFagiGuruKleiRaghSaha2002,
GuptKuma2001, ChenJavdKarbBagnSrinKrau2015, ChenHassKarbKrau2015, Sing2018,
GoelGuhaMuna2006, GuptNaga2013, AdamSvirWard2016, GuptNagaSing2016,
GuptNagaSing2017, GuptJianSing2019}.

A crucial assumption underlying Weitzman's optimality result is that the
rewards in different boxes are independent. This does not always bear out in
practice. Suppose, for example, that you want to buy an item online and look
for a website that offers a cheap price. Your goal is to minimize the price you
pay for the item plus the time it takes to search for a good deal. Since the
websites are competing sellers, it is likely that prices on different sites are
correlated. For another example, consider a route planning service that wants
to determine the fastest route between two destinations from among a set of
potential routes. The driving time for each route is stochastic and depends on
traffic, but the route planning service can obtain its exact value at some
cost. The service wants to minimize the driving time of the route selected plus
the cost spent on querying routes. Once again, because of network effects,
driving times along different routes may be correlated. How do we design an
online search algorithm for these settings?

{\bf In this paper, we provide the first competitive algorithms for Pandora's
Box-type problems with correlations.} We begin our investigation with the
simplest minimization variant of the problem, formalizing the examples
described above: there are $n$ alternatives with unknown costs that are drawn
from some joint distribution. A search algorithm examines these alternatives
one at a time, learning their costs, and after a few steps stops and selects
one of the alternatives. Given sample access to the distribution of costs, our
goal is to develop a search algorithm that minimizes the sum of the expected
cost of the chosen alternative and the number of steps to find it. We call this
the online stochastic search problem.  Henceforth we will refer to the
alternatives as boxes and different instantiations of costs in boxes as
scenarios.

The optimal solution for online stochastic search is a fully-adaptive (FA)
strategy that chooses which box to query each time based on all the costs that
have been observed so far. While these are the best strategies one could hope
for, they are impossible to find or approximate with samples.  For example, it
could be the case that the cost in the first few boxes encode the location of a
box of cost 0 while every other box has infinite cost.  While the best option
can be identified with just few queries, any reasonable approximation to the
optimal cost would need to accurately learn this mapping.  Learning such an
arbitrary mapping however is impossible through samples, unless there is
significant probability of seeing the exact same combination of
costs\footnote{For explicitly given distributions, this is not an issue.
However, this is beyond the scope of the paper as we aim to provide good
strategies that generalize to rich distributions rather than overfitting to and
memorizing the costs in the few scenarios given.}.

\paragraph{Competing against partially adaptive strategies.} Is there any hope
for finding a good strategy for correlated costs? We show that positive results
can be obtained if we target a simpler benchmark. We consider
partially-adaptive (PA) strategies that have a fixed order in which they query
the boxes but may have arbitrarily complex rules on when to stop. 

{\bf Our main positive result is a constant-approximation to the optimal PA
strategy with polynomial time and sample complexity in the number of boxes.} 

Our result directly generalizes the positive results for the special case of
Pandora's box studied in prior work where costs are drawn independently. This
is because optimal strategies for these settings are known to be partially
adaptive intuitively because information about costs of opened boxes does not
imply anything about future boxes. 

In targeting the benchmark of PA strategies, we also give limited power to our
algorithms. All of our approximations are achieved via simple PA strategies
that can be described succinctly. This enables us to learn these strategies
efficiently from data. While one might hope to achieve better results using the
full power of FA strategies, perhaps even surpassing the performance of PA
strategies entirely, we show that a constant factor loss is necessary for
computational reasons.

Our inspiration for using the optimal PA strategy as a benchmark comes from
other contexts where the optimal solution is impossible to approximate well.
One example is prior-free mechanism design where for some objectives such as
revenue, in the absence of stochastic information about input values, no finite
approximation to the optimum can be achieved \cite{hartline2013mechanism}.
Hartline and Roughgarden \cite{HartRoug2008} proposed a template whereby one
characterizes the class of solutions that are optimal under the assumption that
values are drawn i.i.d.~from some unknown distribution. The goal then is to
compete against the best mechanism from this class. Another example is the
concept of static optimality in dynamic data structures. Consider, for example,
the problem of maintaining a binary search tree with the goal of minimizing
search cost over an online sequence of requests. When the requests are drawn
from a fixed distribution, a static search tree is optimal. In the worst case,
however, the optimal in hindsight algorithm maintains a dynamic search tree,
performing rotations between consecutive requests. Achieving constant-factor
competitiveness against the optimal in hindsight solution, known as dynamic
optimality, is a major open problem \cite{SleatorTarjan85,DemaineDynamic07}.
Early work therefore focused on static optimality, or achieving competitiveness
(via a dynamic data structure) against the optimal static tree.  In each of
these cases, an appropriate benchmark is defined by first considering a special
case of the problem (e.g.~i.i.d.~input); characterizing optimal solutions for
that special case; and then competing in the general setting against the best
out of all such solutions. Applying this approach to the online stochastic
search problem, we obtain a benchmark by considering the special case we know
how to solve: namely when the costs in boxes are independently distributed.
Weitzman's work and its generalizations show that in this case the optimal
strategy is always a PA strategy.

\subsection{Results and Techniques}
We now describe our results and techniques in more detail.

\paragraph{Learning a good strategy from data.} To give some intuition about 
why PA strategies are learnable from data, consider the special case where the
costs are either $0$ or $\infty$. Any PA strategy then probes boxes in a
particular order until it finds one with $0$ cost, and then terminates. In
other words, there is only one relevant stopping rule and the space of relevant
PA strategies is ``small'' ($n!$; corresponding to each possible probing
order). This coupled with the boundedness of the objective implies that
$\operatorname{poly}(n)$ samples are enough to find the optimal PA strategy.

The case of general costs is trickier as it is unclear when a low cost option
has been identified. In particular, the class of all PA strategies can be quite
large and complex because the stopping rule can depend in a complex manner on
the costs observed in the boxes. One of our main technical contributions is to
show that once we have determined an order in which to query boxes, it becomes
easy to find an {\em approximately optimal} stopping rule at the loss of a
small constant factor (Lemma~\ref {lem:gen_ski_rental}). This technical lemma
is based on an extension of the ski-rental online
algorithm~\cite{KarlManaMcgeOwic1990} and is presented in
Section~\ref{sec:reduction}. This allows us to focus on finding good
\emph{scenario-aware} strategies -- that is, an ordering of the boxes that
performs well assuming that we know when to stop. The implication is that the
space of ``interesting'' PA strategies is small, characterized by the $n!$
different orderings over boxes, and therefore approximately optimal PA
strategies can be identified from polynomially many samples.

\paragraph{Finding approximately optimal PA strategies.}
As a warm-up, we first develop PA strategies that are competitive against
completely non-adaptive (NA) strategies. NA strategies simply select a fixed
set of boxes to probe and pick the box with the cheapest cost among these.
Despite their simplicity, optimizing over NA strategies using NA strategies is
intractable: it captures the hitting set formulation of set cover and is
therefore hard to approximate better than a logarithmic factor in the number of
scenarios. It is also intractable from the viewpoint of learning: if there is a
tiny probability scenario that has infinite cost on all boxes but one, the
expected cost of the algorithm would be infinite if the algorithm does not
sample that scenario or query all boxes. 

Our first result shows that it is possible to efficiently compute a
scenario-aware PA strategy that beats any NA strategy entirely
(Corollary~\ref{lem:scenario_aware_approx}).  Combining this with our
approximately-optimal stopping rule gives a PA strategy that achieves a
constant factor approximation (1.58) to the optimal NA strategy. While a better
constant factor approximation might be possible through a more direct argument,
we show that it is NP-hard to approximate the optimal NA strategy beyond some
constant (1.278) even if one is allowed to use FA strategies. Our lower-bound
is based on the logarithmic lower-bound for set-cover~\cite{DinuSteu2014} which
restricts how many scenarios can be covered within the first few time steps
(Lemma~\ref{lem:LB_max_coverage}).

Our main result extends the above constant factor approximation guarantees even
against PA strategies. We again restrict our attention to scenario-aware
strategies and seek to find an ordering that approximates the optimal PA
strategy. We solve the resulting problem by formulating a linear programming
relaxation to identify for each scenario a set of ``good'' boxes with suitably
low values. This allows us to reduce the problem at a cost of a constant factor
to finding an ordering of boxes so that the expected time until a scenario
visits one of its ``good'' boxes is minimized. This problem is known as the
min-sum set cover problem and is known to be approximable within a factor of
$4$ \cite{FeigUrieLovaTeta2002}. The resulting approximation factor we obtain
is 9.22. 

\paragraph{Further extensions.} Beyond the problem of identifying a single
option with low cost, we also consider several extensions. One extension is the
case where $k$ options must be identified so that the sum of their costs is
minimized. A further generalization is the case where the set of options must
form a base of rank $k$ in a given underlying matroid. This allows expressing
many combinatorial problems in this framework such as the minimum spanning tree
problem.
For the first extension where any $k$ options are feasible (corresponding to a
uniform matroid) we obtain a constant factor approximation. For general
matroids however, the approximation factor decays to $O( \log k )$. We show
that this is necessary even for the much weaker objective of approximating NA
strategies with arbitrary FA strategies, and even for very simple matroids such
as the partition matroid. We obtain the upper-bounds by modifying the
techniques developed for extensions of min-sum set cover -- the generalized
min-sum set cover and the submodular ranking problem. The following table shows
a summary of the results obtained.

\begin{table}[H]
  \centering
\begin{adjustbox}{center}
\begin{tabular}{|c|c|c|c|}
	\hline & Single Option   &  $k$ Options & Matroid of rank $k$ \\ \hline
	PA vs PA (Upper-bound)  & $9.22$ \textbf{\small [Theorem~\ref{thm:pavspa_ub}]}  & $O(1)$ \textbf{\small [Theorem~\ref{thm:kcoverage}]} & $O( \log k )$ \textbf{\small [Theorem~\ref{thm:matroid_ub}]} \\ \hline
	FA vs NA (Lower-bound) & $1.27$ \textbf{\small [Theorem~\ref{thm:min_LB}]}  & $1.27$ \textbf{\small [Theorem~\ref{thm:min_LB}]} & $\Omega( \log k )$ \textbf{\small [Theorem~\ref{thm:matroid_lb}]} \\ \hline
\end{tabular}
\end{adjustbox}
\caption{The main results shown in this work. The entries correspond to the
achieved competitive ratio for different settings. The upper-bounds are shown
for any PA algorithm compared with the optimal PA. The lower-bounds are shown
in the much weaker setting of FA vs NA.}
\end{table}

While all of the settings above assume that every box takes the same amount of
time to probe (one step), we show in Section~\ref{sec:generalprobingcost} that
our results extend easily to settings where different boxes have different
probing times. We assume that probing times lie in the range $[1,P]$. Both the
running time and sample complexity of our algorithms depend linearly on $P$ and
are efficient when $P$ is polynomially large. This dependence on $P$ for the
sample complexity is necessary to observe scenarios that happen with
probability $O(1/P)$ but contribute a significant amount to the objective.

Finally in Section~\ref{sec:max} we consider a modification of the framework to
maximization instead of minimization problems where the goal is to maximize the
value of the chosen alternative minus the time it takes to find it (as in the
Pandora's Box problem). In contrast to the minimization version, we show that
in this setting even the simplest possible benchmark -- the optimal NA strategy
-- cannot be efficiently approximated within any constant factor using the full
power of FA algorithms.

\subsection{Related Work}


Our framework is inspired by the Pandora's box model which has its roots in the
Economics literature.  Since Weitzman's seminal work on this problem, there has
been a large line of research studying the price of information
\cite{CharFagiGuruKleiRaghSaha2002, GuptKuma2001, ChenJavdKarbBagnSrinKrau2015,
ChenHassKarbKrau2015}  and the structure of approximately optimal rules for
several combinatorial problems~\cite{Sing2018, GoelGuhaMuna2006, GuptNaga2013,
AdamSvirWard2016, GuptNagaSing2016, GuptNagaSing2017, GuptJianSing2019}.  
 
Our work also advances a recent line of research on the foundations of
data-driven algorithm design. The seminal work of Gupta and
Roughgarden~\cite{GuptRoug2016} introduced the problem of algorithm selection
in a distributional learning setting focusing on the number of samples required
to learn an approximately optimal algorithm. A long line of recent research
extends this framework to efficient sample-based optimization over
parameterized classes of algorithms~\cite{AiloChazClarLiuMulzSesh2006,
ClarMulzSesh2012, GuptRoug2016, BalcNagaViteWhit2016, BalcDickSandVite2018,
BalcDickVite2018, KleiLeytLuci2017,
WeisGyorSzep2018,AlabKalaLigeMuscTzamVite2019} In contrast to these results our
work studies optimization over larger, non-parametric classes of algorithms,
indeed any polynomial time (partially-adaptive) algorithm.  Beyond this line of
research, there has also been a lot of work in the context of improving
algorithms using data that combines machine learning predictions to improve
traditional worst case guarantees of online algorithms~\cite{LykoVass2018,
PuroSvitKuma2018,HsuIndyVaki2019,GollPani2019}. 

Finally our work can also be seen as a generalization of the min-sum set cover
problem (MSSC). Indeed MSSC corresponds to the special case where costs are
either $0$ or $\infty$. Some of our LP-rounding techniques are similar to those
developed for MSSC~\cite{FeigUrieLovaTeta2002} and its
generalizations~\cite{AzaGamzIftaYinr2009, BansGuptRavi2010, SkutWill2011,
ImSvirZwaa2012, AzarGamz2010}. Our algorithms for the setting of general
probing times generalize results for the MSSC to settings with arbitrary
``lengths'' for elements.

\section{Model}\label{sec:model}

In the optimal search problem, we are given a set $\boxes$ of $n$ boxes with
unknown costs and a distribution $\dist$ over a set of possible scenarios that
determine these costs. Nature chooses a scenario $s$ from the distribution,
which then instantiates the cost of each box. We use $c_{is}$ to denote the
cost of box $i$ when scenario $s$ is instantiated.

The goal of the online algorithm is to choose a box of small cost while
spending as little time as possible gathering information. The algorithm cannot
directly observe the scenario that is instantiated, however, is allowed to
``probe'' boxes one at a time. Upon probing a box, the algorithm gets to
observe the cost of the box. Formally let $\mathcal{P}_s$ be the random
variable denoting the set of probed boxes when scenario $s$ is instantiated and
let $i_s\in \mathcal{P}_s$ be the (random) index of the box chosen by the
algorithm. We require $i_s\in\mathcal{P}_s$, that is, the algorithm must probe
a box to choose it. Note that the randomness in the choice of $\mathcal{P}_s$
and $i_s$ arises both from the random instantiation of scenarios as well as
from any coins the algorithm itself may flip. Our goal then is to minimize the
total probing time plus the cost of the chosen box:
\[
	\E{s}{\min_{i\in \mathcal{P}_s} c_{is} + |\mathcal{P}_s|}.
\]

Any online algorithm can be described by the pair $(\sigma, \tau)$, where
$\sigma$ is a permutation of the boxes representing the order in which they get
probed, and $\tau$ is a stopping rule -- the time at which the algorithm stops
probing and returns the minimum cost it has seen so far. Observe that in its
full generality, an algorithm may choose the $i$'th box to probe, $\sigma(i)$,
as a function of the identities and costs of the first $i-1$ boxes,
$\{\sigma(1), \cdots, \sigma(i-1)\}$ and $\{c_{\sigma(1)s}, \cdots,
		c_{\sigma(i-1)s}\}$\footnote{For some realized scenario $s\in
		\scenario$.}.  Likewise, the decision of setting $\tau=i$ for $i\in
		[n]$ may depend on $\{\sigma(1), \cdots, \sigma(i)\}$ and
		$\{c_{\sigma(1)s}, \cdots, c_{\sigma(i)s}\}$. Optimizing over the class
		of all such algorithms is intractable. So we will consider simpler
		classes of strategies, as formalized in the following definition.

\begin{definition}[Adaptivity of Strategies]
	In a \emph{Fully-Adaptive (FA)} strategy, both $\sigma$ and 
	$\tau$ can depend on any costs seen in a previous time step,
        as described above. 

	In a \emph{Partially-Adaptive (PA)} strategy, the sequence
        $\sigma$ is independent of the costs observed in probed
        boxes. The sequence is determined before any boxes are
        probed. However, the stopping rule $\tau$ can depend on the
        identities and costs of boxes probed previously.

	In a \emph{Non-Adaptive (NA)} strategy, both $\sigma$ and
        $\tau$ are fixed before any costs are revealed to the
        algorithm. In particular, the algorithm probes a fixed subset
        of the boxes, $I\subseteq [n]$, and returns the minimum cost
        $\min_{i\in I} c_{is}$. The algorithm's expected total cost is then
        $\E{s}{\min_{i\in I} c_{is} + |I|}$.
\end{definition}	

\paragraph{General feasibility constraints.} In
Section~\ref{sec:extensions} we study extensions of the search problem where
our goal is to pick multiple boxes satisfying a given feasibility constraint.
Let $\feas\subseteq 2^{\boxes}$ denote the feasibility constraint. Our goal is
to probe boxes in some order and select a subset of the probed boxes that is
feasible. Once again we can describe an algorithm using the pair $(\sigma,
\tau)$ where $\sigma$ denotes the probing order, and $\tau$ denotes the
stopping time at which the algorithm stops and returns the cheapest feasible
set found so far. The total cost of the algorithm then is the cost of the
feasible set returned plus the stopping time. We emphasize that the algorithm
faces the same feasibility constraint in every scenario. We consider two
different kinds of feasibility constraints. In the first, the algorithm is
required to select exactly $k$ boxes for some $k\ge 1$. In the second, the
algorithm is required to select a basis of a given matroid.

\section{A reduction to scenario-aware strategies and its implications
to learning}
\label{sec:reduction}

Recall that designing a PA strategy involves determining a non-adaptive probing
order, and a good stopping rule for that probing order. We do not place any
bounds on the number of different scenarios, $m$, or the support size and range
of the boxes' costs. These numbers can be exponential or even unbounded. As a
result, the optimal stopping rule can be very complicated and it appears to be
challenging to characterize the set of all possible PA strategies. We simplify
the optimization problem by providing {\em extra power} to the algorithm and
then removing this power at a small loss in approximation factor. 

In particular, we define a \emph{Scenario-Aware Partially-Adaptive (SPA)}
strategy as one where the probing order $\sigma$ is independent of the costs
observed in probed boxes, however, the stopping time $\tau$ is a function of
the instantiated scenario $s$. In other words, the algorithm fixes a probing
order, then learns of the scenario instantiated, and then determines a stopping
rule for the chosen probing order based on the revealed scenario.

Observe that once a probing order and instantiated scenario are fixed, it is
trivial to determine an optimal stopping time in a scenario aware manner. The
problem therefore boils down to determining a good probing order. The space of
all possible SPA strategies is also likewise much smaller and simpler than the
space of all possible PA strategies. We can therefore argue that in order to
learn a good SPA strategy, it suffices to optimize over a small sample of
scenarios drawn randomly from the underlying distribution. We denote the cost
of an SPA strategy with probing order $\sigma$ by $\cost(\sigma)$.

On the other hand, we argue that scenario-awareness does not buy much power for
the algorithm. In particular, given any fixed probing order, we can construct a
stopping time that depends only on the observed costs, but that achieves a
constant factor approximation to the optimal scenario-aware stopping time for
that probing order.
 
The rest of this section is organized as follows. In
Section~\ref{subsec:ski-rental} we exhibit a connection between our problem and
a generalized version of the ski rental problem to show that PA strategies are
competitive against SPA strategies. In Section~\ref{subsec:learning} we show
that optimizing for SPA strategies over a small sample of scenarios suffices to
obtain a good approximation. In Section~\ref{subsec:LP} we develop LP
relaxations for the optimal NA and SPA strategies. Then in the remainder of the
paper we focus on finding approximately-optimal SPA strategies over a small set
of scenarios.

\subsection{Ski Rental with varying buy costs}
\label{subsec:ski-rental}

We now define a generalized version of the ski rental problem which is closely
related to SPA strategies. The input to the generalized version is a sequence
of non-increasing buy costs, $a_1\ge a_2\ge a_3 \ge \ldots$.  These costs are
presented one at a time to the algorithm. At each step $t$, the algorithm
decides to either rent skis at a cost of $1$, or buy skis at a cost of $a_t$.
If the algorithm decides to buy, then it incurs no further costs for the
remainder of the process.  Observe that an offline algorithm that knows the
entire cost sequences $a_1, a_2, \ldots$ can pay $\min_{ t\geq 1} (t-1 + a_t)$.
We call this problem \emph{ski rental with time-varying buy costs}. The
original ski rental problem is the special case where $a_t=B$ or $0$ from the
time we stop skiing and on.

We first provide a simple randomized algorithm for ski rental with time-varying
costs that achieves a competitive ratio of $e/(e-1)$. Then we extend this to
general $p_t$ in Corollary~\ref{cor:varying_rent}. Our algorithm uses the
randomized algorithm of \cite{KarlManaMcgeOwic1990} for ski rental as a
building block, essentially by starting a new instance of ski rental every time
the cost of the offline optimum changes. The full proof of this result is
included in Section~\ref{sec:lemma_proof} of the appendix.

\begin{lemma}[Ski Rental with time-varying buy costs]\label{lem:gen_ski_rental}
  Consider any sequence $a_1\ge a_2\ge \ldots$. There exists an online algorithm that
	chooses a stopping time $t$ so that
	\[ t-1 + a_{t} \le \frac{e}{e-1} \min_j \{ j-1+a_j\}. \]
\end{lemma}

The next corollary connects scenario-aware partially-adaptive strategies with
partially-adaptive strategy through our competitive algorithm for ski-rental
with time-varying costs. Specifically, given an SPA strategy, we construct an
instance of the ski-rental problem, where the buy cost $a_t$ at any step is
equal to the cost of the best feasible solution seen so far by the SPA
strategy. The rent cost of the ski rental instance reflects the probing time of
the search algorithm, whereas the buy cost reflects the cost of the boxes
chosen by the algorithm. Our algorithm for ski rental chooses a stopping time
as a function of the costs observed in the past and without knowing the
(scenario-dependent) costs to be revealed in the future, and therefore gives us
a PA strategy for the search problem.  
 
This result is formalized below; the proof can be found in
Section~\ref{sec:lemma_proof} of the appendix.

\begin{corollary}\label{lem:scenario_aware_approx}
  Given any \emph{scenario-aware partially-adaptive} strategy $\sigma$,
  we can efficiently construct a stopping time $\tau$, such that the
  cost of the \emph{partially-adaptive} strategy $(\sigma, \tau)$ is
  no more than a factor of $e/(e-1)$ times the cost of $\sigma$.
\end{corollary}

\subsection{Learning a good probing order}
\label{subsec:learning}
Henceforth, we focus on designing good scenario-aware partially adaptive
strategies for the search problem. As noted previously, once we fix a probing
order, determining the optimal scenario-aware stopping time is easy. We will
now show that in order to optimize over all possible probing orders, it
suffices to optimize with respect to a small set of scenarios drawn randomly
from the underlying distribution.

Formally, let $\dist$ denote the distribution over scenarios and let
$\scenario$ be a collection of $m$ scenarios drawn independently from $\dist$,
with $m$ being a large enough polynomial in $n$. Then, we claim that with high
probability, for {\em every} probing order $\sigma$, $\cost_\dist(\sigma)$ is
close to $\cost_\scenario(\sigma)$, where $\cost_\dist(\sigma)$ denotes the
total expected cost of the SPA strategy $\sigma$ over the scenario distribution
$\dist$, and $\cost_\scenario(\sigma)$ denotes its cost over the uniform
distribution over the sample $\scenario$. The implication is that it suffices
for us to optimize for SPA strategies over scenario distributions with finite
small support. 

\begin{lemma}
\label{lem:learn}
Let $\epsilon, \delta>0$
be given parameters. Let $\scenario$ be a set of $m$ scenarios chosen
independently at random from $\dist$ with
$m = \poly(n,1/\epsilon, \log(1/\delta))$. Then, with probability at
least $1-\delta$, for all permutations $\pi:[n]\rightarrow [n]$, we
have
  \[\cost_\scenario(\pi)\in (1\pm \epsilon) \cost_\dist(\pi).\]
\end{lemma}

\begin{proof}
  Fix a permutation $\pi$. For scenario $s$, let $\cost_s(\pi) = \min_i \{i+
  c_{\pi(i) s}\}$ denote the total cost incurred by SPA strategy $\pi$ in
  scenario $s$. Observe that for any $\pi$ and any $s$, we have $\cost_s(\pi)
  \in [1+\min_i c_{is},n+\min_i c_{is}]$. Furthermore, $\cost_\dist(\pi) =
  \E{s\sim \dist}{\cost_s(\pi)}$, and $\cost_\scenario(\pi) =
  \frac{1}{|\scenario|}\sum_{s\in\scenario} cost_s(\pi)$. The lemma now follows
  by using the Hoeffding inequality and applying the union bound over all
  possible permutations $\pi$.
\end{proof}

Combining Corollary~\ref{lem:scenario_aware_approx} and Lemma~\ref{lem:learn}
yields the following theorem.

\begin{theorem}\label{thm:PA-to-SPA-reduction}
	Suppose there exists an algorithm for the optimal search problem that runs
	in time polynomial in the number of boxes $n$ and the number of scenarios
	$m$, and returns an SPA strategy achieving an $\alpha$-approximation. Then,
	for any $\epsilon>0$, there exists an algorithm that runs in time
	polynomial in $n$ and $1/\epsilon$ and returns a PA strategy with
	competitive ratio $\frac{e}{e-1}(1+\epsilon)\alpha$, where $n=|\boxes|$. 
\end{theorem}

\subsection{LP formulations}\label{subsec:LP}
We will now construct an LP relaxation for the optimal scenario-aware partially
adaptive strategy. Following Theorem~\ref{thm:PA-to-SPA-reduction} we focus on
the setting where the scenario distribution is uniform over a small support set
$\scenario$.

The program~\eqref{lp-spa} is given below and is similar to the one used for
the generalized min-sum set cover problem in \cite{BansGuptRavi2010} and
\cite{SkutWill2011}. Denote by $\timeset$ to set of time steps.  Let $x_{it}$
be an indicator variable for whether box $i$ is opened at time $t$.
Constraints~\eqref{eq:mLP_PA_1_box_at_a_time} and
\eqref{eq:mLP_PA_each_box_once} model matching constraints between boxes and
time slots. The variable $z_{ist}$ indicates whether box $i$ is selected for
scenario $s$ at time $t$. Constraints~\eqref{eq:mLP_PA_select_opened} ensure
that we only select opened boxes.
Constraints~\eqref{eq:mLP_PA_one_value_per_scen} ensure that for every scenario
we have selected exactly one box. The cost of the box assigned to scenario $s$
is given by $\sum_{i,t} z_{ist} c_{is}$. Furthermore, for any scenario $s$ and
time $t$, the sum $\sum_i z_{ist}$ indicates whether the scenario is covered at
time $t$, and therefore, the probing time for the scenario is given by $\sum_t
\sum_i tz_{ist}$. 

  \begin{align}
    & \text{minimize} & \frac{1}{|\scenario|} \sum_{i\in\boxes, s\in
    \scenario, t\in \timeset}tz_{ist} & + \frac{1}{|\scenario|}\sum_{i\in
                                   \boxes ,s\in \scenario,t\in \timeset}
                                   c_{is}z_{ist} &  \tag{LP-SPA} \label{lp-spa}\\
    & \text{subject to} & \sum_{i\in \boxes}x_{it} & = 1, & \forall t\in \timeset \label{eq:mLP_PA_1_box_at_a_time}\\
    & & \sum_{t\in \timeset}x_{it}  & \leq 1, & \forall i\in \boxes \label{eq:mLP_PA_each_box_once} \\
    & & z_{ist} & \leq x_{it}, & \forall s\in \scenario,i\in \boxes,t\in\timeset \label{eq:mLP_PA_select_opened}\\
    & & \sum_{t'\in\timeset, i\in \boxes}z_{ist'} & = 1, & \forall s\in \scenario \label{eq:mLP_PA_one_value_per_scen}\\
    & & x_{it}, z_{ist} & \in [0,1]  & \forall s\in \scenario,i \in \boxes, t\in \timeset \notag
  \end{align}

As a warm-up for our main result, we approximate the optimal NA strategy by a
PA strategy. The relaxation~ \eqref{lp-na} for the optimal NA strategy is
simpler. Here $x_i$ is an indicator variable for whether box $i$ is opened and
$z_{is}$ indicates whether box $i$ is assigned to scenario $s$.

\begin{alignat}{3}
  \text{minimize}  \quad & \sum_{i\in \boxes} x_i   & \quad  + \quad
  & \frac{1}{|\scenario|}\sum_{i\in \boxes,s\in \scenario}c_{is}z_{is}
  & \tag{LP-NA} \label{lp-na}\\
  \text{subject to}\quad & \sum_{i\in \boxes}z_{is} & \quad  = \quad   & 1, 				& \forall s\in \scenario \label{eq:mLP_NA_select_1}\\
						 & \hspace{0.7cm} z_{is} 			 & \quad \leq \quad & x_i , & \forall i\in \boxes, s\in \scenario \notag\label{eq:mLP_NA_select_opened}\\
						 &  \hspace{0.2cm} x_i, z_{is} 		 &\quad \in \quad & [0,1] & \forall i\in \boxes, s\in \scenario \notag
\end{alignat}

\section{Competing with the non-adaptive benchmark}\label{sec:pavsna}

As a warm-up to our main result, in this section we consider competing against
the optimal non-adaptive strategy. Recall that a non-adaptive strategy probes a
fixed subset of boxes, and then picks a probed box of minimum cost. Is it
possible to efficiently find an adaptive strategy that performs just as well?
We show two results. On the one hand, in Section~\ref{sec:pavsna-ub} we show
that we can efficiently find an SPA strategy that beats the performance of the
optimal NA strategy. This along with Theorem~\ref{thm:PA-to-SPA-reduction}
implies that we can efficiently find an $e/(e-1)\approx 1.582$-competitive PA
strategy. On the other hand, in Section~\ref{sec:favsna} we show that it is
NP-hard to obtain a competitive ratio better than $1.278$ against the optimal
NA strategy even using the full power of FA strategies.

\subsection{An upper bound via PA strategies}
\label{sec:pavsna-ub}

Our main result of this section is as follows. 

\begin{lemma}\label{lem:pavsna_exists_sequence}
	We can efficiently compute a scenario-aware partially-adaptive strategy
	with competitive ratio $1$ against the optimal non-adaptive strategy. 
\end{lemma}

\noindent
Putting this together with Theorem~\ref{thm:PA-to-SPA-reduction} we
get the following theorem.

\begin{theorem}\label{thm:pavsna_1.58_approx}
  We can efficiently find a partially-adaptive strategy with total expected cost at 
  most $e/(e-1)$ times the total cost of the optimal non-adaptive strategy. 
\end{theorem}

\begin{proof}[Proof of Lemma~\ref{lem:pavsna_exists_sequence}]
  We use the LP relaxation~\eqref{lp-na} from Section~\ref{sec:model}. Given an
  optimal fractional solution $(\bm{x}, \bm{z})$, we denote by
  $\opt_{c,s}=\sum_{i}c_{is}z_{is}$ the cost for scenario $s$ in this solution,
  and by $\opt_c = \frac{1}{|\scenario|}\sum_{s}\opt_{c,s}$ the cost for all
  scenarios. Let $\opt_t=\sum_{i\in \boxes}x_i$ denote the probing time for the
  fractional solution. Similarly, we define $\alg_t$, $\alg_c$ and $\alg_{c,s}$
  to be the algorithm's query time, cost for all scenarios and cost for
  scenario $s$ respectively.

  Algorithm~\ref{alg:PAvsNA_1_cov} rounds $(\bm{x}, \bm{z})$ to an SPA
  strategy. Note that the probing order $\sigma$ in the rounded solution is
  independent of the instantiated scenario, but the stopping time $\tau_s$
  depends on the scenario specific variables $z_{is}$. $\tau_s$ is not
  necessarily the optimal stopping time for the constructed probing order, but
  its definition allows us to relate the cost of our solution to the fractional
  cost $\opt_c$.

	\begin{algorithm}[H]
		\caption{SPA vs NA}\label{alg:PAvsNA_1_cov}
		\KwData{Solution $\bm{x}, \bm{z}$ to program~\eqref{lp-na}; scenario $s$}
		$\sigma:=$ For $t\geq 1$, select and open box $i$
                with probability $\frac{x_i}{\sum_{i\in \boxes}x_i}$.\\
		$\tau_s:=$ If box $i$ is opened at step $t$, select the box and stop with probability $\frac{z_{is}}{x_i}$.
	\end{algorithm}

Notice that for each step $t$, the probability of stopping is  
\[
	\Pr{\text{stop at step }t} 
		= \sum_{i\in \boxes}\frac{x_i}{\sum_{i\in \boxes}x_i}\frac{z_{is}}{x_i} 
		= \frac{\sum_{i\in \boxes }z_{is}}{\sum_{i\in \boxes}x_i} 
		= \frac{1}{\opt_t},
\]
where we used the first set of LP constraints~\eqref{eq:mLP_NA_select_1} and
the definition of $\opt_t$.  Observe that the probability is independent of the
step $t$ and therefore $\E{}{\alg_t} = \opt_t$.  The expected cost of the
algorithm is
\begin{align*}
		\E{}{\alg_{c,s}}  &= \sum_{i\in \boxes, t} \Pr{\text{select $i$ at step $t$ }| \text{ stop at step $t$}} \Pr{\text{stop at step }t} c_{is} \\
			  &\leq \sum_{i\in \boxes,t} \frac{z_{is}}{\sum_{i\in \boxes}z_{is}}\Pr{\text{stop at step }t} c_{is}
				= \sum_{i\in \boxes} z_{is} c_{is}
				= \opt_{c,s}
\end{align*}
Taking expectation over all scenarios we get
$\E{}{\alg_{c}}\leq\opt_{c}$, and the lemma follows.

\end{proof}

\subsection{A lower bound for FA strategies}
\label{sec:favsna}

We now show that we cannot achieve a competitive ratio of $1$ against
the optimal NA strategy even if we use the full power of fully
adaptive strategies.
\begin{theorem}\label{thm:min_LB}
  Assuming P$\neq$NP, no computationally efficient fully-adaptive algorithm can
  approximate the optimal non-adaptive strategy within a factor smaller than
  $1.278$.
\end{theorem}

Our lower bound is based on the hardness of approximating Set
Cover. We use the following lemma which rules out bicriteria results
for Set Cover; a proof can be found in
Appendix~\ref{sec:apndx_minLB}. 

\begin{lemma}\label{lem:LB_max_coverage}
  Unless P=NP,  for any constant $\epsilon>0$, there is no algorithm that for
  every instance of Set Cover finds $k$ sets that cover at least
  $1-\lp(1-\frac{1+\e}{\opt}\rp)^k$ of the elements for some integer $k \in
  \lp[ 1, \frac{\log n}{1+\e}\opt \rp]$.
\end{lemma}

\begin{proof}[Proof of Theorem~\ref{thm:min_LB}]
  Let $H>0$ and $p\in [0,1]$ be appropriate constants, to be
  determined later. We will define a family of instances of the
  optimal search problem based on set cover. Let
  $\mathcal{SC} = ([m], \{S_1, \ldots, S_n\})$ be a set cover instance
  with $m$ elements and $n$ sets. Denote its optimal value by
  $\opt_{SC}$.  To transform this into an instance of the search
  problem, every element $e_j \in [m]$ corresponds to a scenario $j$,
  and every set $S_i$ to a box $i$. We set $c_{ij}=0$ iff
  $e_j \in S_i$, otherwise $c_{ij} = H$. We also add a new scenario
  $X$ with $v_{Xi} = H,\ \forall i\in [n]$. Scenario $X$ occurs with
  probability $p$ and all the other $m$ scenarios happen with
  probability $(1-p)/m$ each.

  In this instance, the total cost of optimal non-adaptive strategy is
  $\opt_{NA} \leq pH + \opt_{SC}$, since we may pay the set-cover cost
  to cover all scenarios other than $X$, and pay an additional cost
  $H$ to cover $X$.

  Consider any computationally efficient algorithm $\al$ that returns a
  fully adaptive strategy for such an instance. Since the costs of the
  boxes are $0$ or $H$, we may assume without loss of generality that
  any FA strategy stops probing as soon as it observes a box with cost
  $0$ and chooses that box. We say that the strategy covers a scenario
  when it finds a box of cost $0$ in that scenario. Furthermore, prior
  to finding a box with cost $0$, the FA strategy learns no
  information about which scenario it is in other than that the
  scenario is as yet uncovered. Consequently, the strategy follows a
  fixed probing order that is independent of the scenario that is
  instantiated. We can now convert such a strategy into a bicriteria
  approximation for the underlying set cover instance. In particular,
  for $k\in [n]$, let $r_k$ denote the number of scenarios that are
  covered by the first $k$ probed boxes. Then, we obtain a solution to
  $\mathcal{SC}$ with $k$ sets covering $r_k$ elements. By
  Lemma~\ref{lem:LB_max_coverage} then, for every $\e>0$, there must
  exist an instance of set cover, $\mathcal{SC}$, and by extension an
  instance of optimal search, on which $\al$ satisfies
  $r_k\le 1-\lp(1-\frac{1+\e}{\opt_{SC}}\rp)^{k-1}$ for all
  $k\le \frac{\log n}{(1+\e)}\opt_{SC}$.

  For the rest of the argument, we focus on that hard instance for
  $\al$. Let $N$ denote the maximum number of boxes $\al$ probes before
  stopping to return a box of cost $H$.\footnote{We may safely assume
    that $N\le\log n\opt_{SC}$.}
 
  Then the expected query time of the strategy is at least
	\begin{align}
		\Pr{s=X}\cdot N & + \Pr{s\neq X}\sum_{k=1}^{N}\Pr{\text{FA reaches step }k|s\neq X}\nonumber\\
					    &\geq pN+(1-p)\sum_{k=1}^{N}\lp( 1-\frac{1+\e}{\opt_{SC}} \rp)^{k-1}\nonumber\\
						& = pN+(1-p)\lp( 1 - \lp(1-\frac{1+\e}{\opt_{SC}}\rp)^N \rp)\frac{\opt_{SC}}{1+\e}.\label{eqn:lb_probingcost}
	\end{align}

\noindent
	On the other hand, the expected cost of the FA strategy is at least
	\[
		H(\Pr{s=X}+\Pr{s\neq X\wedge \text{FA didn't find
                    cost 0 in first } N \text{ steps }})\geq pH+(1-p)H\left(1-\frac{1+\e}{\opt_{SC}}\right)^{N}.
	\]
	
\noindent
	Thus the total cost of such fully-adaptive strategy is lower bounded by
	\begin{equation*}
		\alg_{FA} \geq  pH  
		+ (1-p) H \lp( 1 - \frac{1 + \e}{\opt_{SC}} \rp)^N + pN 
			+ (1-p) \lp( 1 -
			\lp(1-\frac{1+\e}{\opt_{SC}}\rp)^N\rp)\frac{\opt_{SC}}{1+\e}.
	\end{equation*}

\noindent
Let $x$ be defined so that
$( 1 - \frac{1 + \e}{\opt_{SC}})^N =e^{-x}$. Then,
$N=-x/\ln(1-\frac{1 + \e}{\opt_{SC}})\ge x(\frac{\opt_{SC}}{1 +
  \e}-1)$. Substituting these expressions in the above equation we get
	\begin{equation*}
		\alg_{FA} \ge pH 
		+ (1-p) He^{-x}
		+ p\cdot x \lp(\frac{\opt_{SC}}{1+\e} -1\rp)
		+ (1-p)( 1 - e^{-x})\frac{\opt_{SC}}{1+\e}.
	\end{equation*}

        The RHS is minimized at
        $x = \ln \lp(
        \frac{(1-p)(H(1+\e)-\opt_{SC})}{p(\opt_{SC}-(1+\e))}\rp)$. By
        setting $\epsilon\to 0$, $p=0.22$ and $H=4.59 \opt_{SC}$, the
        competitive ratio becomes
	 \begin{align*}
			 \frac{\alg_{FA}}{\opt_{NA}} 
			 \geq 1.278
	 \end{align*}
   when $\opt_{SC}\to\infty$.
\end{proof}

\section{Competing with the partially-adaptive benchmark}\label{sec:pavspa}
Moving on to our main result, in this section we compete against the optimal
partially-adaptive strategy. Recall that the program~\eqref{lp-spa} is a
relaxation for the optimal SPA strategy, and therefore, also bounds from below
the cost of the optimal PA strategy. We round the optimal solution to this LP
to obtain a constant-competitive SPA strategy.

Given a solution to \eqref{lp-spa}, we identify for each scenario a subset of
low cost boxes. Our goal is then to find a probing order, so that for each
scenario we quickly find one of the low cost boxes. This problem of
``covering'' every scenario with a low cost box is identical to the min-sum set
cover (MSSC) problem introduced by \cite{FeigUrieLovaTeta2002}. Employing this
connection allows us to convert an approximation for MSSC into an SPA strategy
at a slight loss in approximation factor.

Our main result is as follows.

\begin{lemma}\label{lem:pavspa_scen_aware}
	There exists a scenario-aware partially-adaptive strategy with competitive
	ratio $3+2\sqrt{2}$ against the optimal partially-adaptive strategy.
\end{lemma}

Combining this with Theorem~\ref{thm:PA-to-SPA-reduction} we get the following
theorem.

\begin{theorem}\label{thm:pavspa_ub}
  We can efficiently find a partially-adaptive strategy that is
  $(3 + 2\sqrt{2}) \frac{e}{e-1}=9.22$-competitive against the optimal
  partially-adaptive strategy.
\end{theorem}

\begin{proof}[Proof of Lemma~\ref{lem:pavspa_scen_aware}]
	We use the LP formulation \ref{lp-spa} from Section~\ref{sec:model}. Recall
	that $x_{it}$ denotes the extent to which box $i$ is opened at time $t$,
	$z_{ist}$ denotes the extent to which box $i$ is chosen for scenario $s$ at
	time $t$.

	 As mentioned previously, we will employ a $4$-approximation to the MSSC by
	 \cite{FeigUrieLovaTeta2002} in our algorithm. The input to MSSC is an
	 instance of set cover. In our context, the sets of the set cover are boxes
	 and each scenario $s$ has an element $L_s$ corresponding to it. The goal
	 is to find an ordering $\sigma$ over the sets/boxes so as to minimize the
	 sum of the cover times of the elements/scenarios, where the cover time of
	 an element is the index of the first set in $\sigma$ that contains it. The
	 following is an LP relaxation for MSSC; observe its similarity to
	 \eqref{lp-spa}.  \cite{FeigUrieLovaTeta2002} provide a greedy algorithm
	 that $4$-approximates the optimal solution to this LP.

  \begin{align}
    & \text{minimize} & \frac{1}{|\scenario|} \sum_{i\in\boxes, s\in
    \scenario, t\in \timeset}tz_{ist}  \tag{LP-MSSC} \label{lp-mssc}\\
    & \text{subject to} &
                          \eqref{eq:mLP_PA_1_box_at_a_time}-\eqref{eq:mLP_PA_select_opened}
                          \notag \\
    & & \sum_{t'\in\timeset, i\in L_s}z_{ist'} & \ge 1, & \forall s\in
                                                     \scenario \notag \\
    & & x_{it}, z_{ist} & \in [0,1]  & \forall s\in \scenario,i \in \boxes, t\in \timeset \notag
  \end{align}

  Define $\alpha = 3 + 2 \sqrt{2}$. Given an optimal solution $\mathcal{I} =
  (\bm{x}, \bm{z})$ to \eqref{lp-spa}, we will now construct an instance
  $\mathcal{I}'$ of MSSC (by specifying the elements $L_s$\footnote{Each
  element/scenario can be thought of as a set of the boxes/sets that cover
  it.}) with the following properties:
	\begin{enumerate}[label=(\roman*)]
        \item There exists an integral solution $\sigma$ for
          $\mathcal{I}'$ with cover time at most $\alpha$ times the
          query time for $(\bm{x}, \bm{z})$.
        \item Any integral solution $\sigma$ for $\mathcal{I}'$ can be
          paired with an appropriate stopping time $\tau_s$, so that
          the query time of $(\sigma,\tau_s)$ is at most the
          MSSC cover time of $\sigma$, and the cost of
          $(\sigma,\tau_s)$ is at most $\alpha$ times the fractional
          cost for $(\bm{x}, \bm{z})$.
	\end{enumerate}

	\paragraph{Constructing a ``good" $\mathcal{I}'$:}
        For each scenario $s$, we define a set of ``low'' cost boxes as 
	\[ 
		L_s = \{i: c_{is} \leq \alpha\opt_{c,s}^{\mathcal{I}}\}.
	\]
		The second property above is immediate from this definition. In
		particular, we define the stopping time $\tau_s$ as the first time we
		encounter a box $i\in L_s$. 

	 For property (i), we first show that instance $\mathcal{I}'$ admits a good
	 fractional solution. 

	  While \eqref{lp-spa} allows assigning any arbitrary boxes to a scenario,
	  \eqref{lp-mssc} requires assigning only the boxes in $L_s$ to scenario
	  $s$.  In order to convert this into a feasible solution to
	  \eqref{lp-mssc}, we first scale up all of the variables by a factor of
	  $\frac{\alpha}{\alpha-1}$. Specifically, set $\bm{x}' =
	  \frac{\alpha}{\alpha-1} \bm{x}$; $z'_{ist}=\frac{\alpha}{\alpha-1}
	  z_{ist}$ for all $s, t, i\in s$; and $z'_{ist} = 0$ for all $s, t,
	  i\not\in s$.  Now we need to ensure that all the constraints of
	  \eqref{lp-mssc} are satisfied.

	Observe initially that by Markov's inequality, for all $s$, $\sum_{t, i\in
	L_s} z_{ist}\ge 1-1/\alpha$. Therefore, by scaling the $z'_{ist}$ as above
	we have that $\sum_{it} z'_{ist}\ge 1$ for all $s$.  To fix
	\eqref{eq:mLP_PA_each_box_once}, if for some $i\in\boxes$ we have $\sum_t
	x'_{it}>1$, let $t'$ be the smallest time at which $\sum_{t\le t'}
	x'_{it}>1$. We set $x'_{it}=0$ for all $t>t'$ and $x'_{it'} = 1-\sum_{t<
	t'} x'_{it}$.  Likewise, modify $\bm{z}'$ so as to achieve
	\eqref{eq:mLP_PA_select_opened} as well as ensure that every variable lies
	in $[0,1]$.

	It remains to argue that constraints~\eqref{eq:mLP_PA_1_box_at_a_time} can
	be fixed at a small extra cost.  Observe that for any $t$, $\sum_{i\in
	\boxes} x'_{it}\le \frac{\alpha}{\alpha-1}$. We therefore ``dilate'' time
	by a factor of $\frac{\alpha}{\alpha-1}$ in order to accommodate the higher
	load. Formally, interpret $x$, $x'$, $z$ and $z'$ as continuous step
	functions of $t$. Then the objective function of \eqref{lp-mssc} can be
	written as $\sum_{s,i}\int_{t=0}^{t=n} \lceil t\rceil z'_{ist}
	dt$\footnote{Note that we are not adding extra time steps, and therefore
	the cost of the objective does not change.}.  Dilating time by a factor of
	$\frac{\alpha}{\alpha-1}$ gives us the objective $\frac{1}{|\scenario|}
	\sum_{s,i}\int_{t=0}^{t=n} \lp\lceil t\frac{\alpha}{\alpha-1}\rp\rceil
	z'_{ist} dt$.

	Since $z'_{ist}\leq \frac{\alpha}{\alpha-1} z_{ist}$ for any $i,s,t$, the
	expected query time is upper bounded by
        \begin{align*}
			& \frac{1}{|\scenario|}  \sum_{i,s}\int_{t=0}^{t=n} \lp\lceil
			t\frac{\alpha}{\alpha-1}\rp\rceil \cdot \frac{\alpha}{\alpha-1}z_{ist}\, dt\\
                        & \leq \frac{1}{|\scenario|} \sum_{i,s,t}\lp( \frac{\alpha}{\alpha-1} \rp)^2 t z_{ist} \\
                        & = \lp(\frac{\alpha}{\alpha-1}\rp)^2\cdot\text{
                          Query time of } (\bm{x}, \bm{z})
        \end{align*}
	where for the second inequality we used the following
	Lemma~\ref{lem:integral_bound} with $\beta=\alpha/(\alpha-1)$. The proof of
	the lemma is deferred to Section~\ref{sec:apndx_pavspa} of the appendix.

	\begin{lemma}\label{lem:integral_bound}
	For any $\beta > 1$,
		\[ 
			\int_{t-1}^{t}\lceil\beta t'\rceil dt' \leq \beta t.
		\]
	\end{lemma}

        \paragraph{Applying greedy algorithm for min-sum set cover.}
	We have so far constructed a new instance $\mathcal{I}'$ of
        MSSC along with a feasible fractional solution
        $(\bm{x'}, \bm{z'})$ with cover time at most
        $\alpha^2/(\alpha-1)^2$ times the query time for
        $(\bm{x}, \bm{z})$. The greedy algorithm of
        \cite{FeigUrieLovaTeta2002} finds a probing order over the
        boxes with query time at most $4$ times the cover time of
        $(\bm{x'}, \bm{z'})$, that is, at most
        $4\alpha^2/(\alpha-1)^2=\alpha$ times the query time for
        $(\bm{x}, \bm{z})$, where the equality follows from the
        definition of $\alpha$. Property (1) therefore holds and the
        lemma follows.
	
\end{proof}

\section{Extension to other feasibility constraints}\label{sec:extensions}
In this section we extend the problem in cases where there is a feasibility
constraint $\mathcal{F}$, that limits what or how many boxes we can
choose. We consider the cases where we are required to select $k$ distinct boxes,
and $k$ independent boxes from a matroid. In both cases we design SPA strategies
that can be converted to PA. These two variants are described in more
detail in subsections \ref{subsec:generalK} and \ref{subsec:matroids} that
follow.

\subsection{Selecting \texorpdfstring{$k$}{} items}
\label{subsec:generalK}

In this section $\mathcal{F}$ requires that we pick $k$ boxes to minimize the
total cost and query time. As in Section~\ref{sec:pavspa} we aim to compete
against the optimal partially-adaptive strategy. We design a PA strategy which
achieves an $O(1)$-competitive ratio. If $c_{is}\in \{0, \infty\}$, the problem
is the generalized min-sum set cover problem first introduced in
\cite{AzaGamzIftaYinr2009}. \cite{AzaGamzIftaYinr2009} gave a $\log
n$-approximation, which then was improved to a constant in
\cite{BansGuptRavi2010} via an LP-rounding based algorithm.  Our proof will
follow the latter proof in spirit, and generalize to the case where boxes have
arbitrary values.  Our main result is the following.

\begin{lemma}\label{lem:scen_aware_generalk_UB}
	There exists a scenario-aware partially-adaptive $O(1)$-competitive
	algorithm to the optimal partially-adaptive algorithm for picking $k$ boxes.
 \end{lemma}
\noindent
 Combining this with Theorem~\ref{thm:PA-to-SPA-reduction} we get the following
 theorem.

\begin{theorem}\label{thm:kcoverage}
  We can efficiently find a partially-adaptive strategy for optimal
  search with $k$ options that is $O(1)$-competitive against the
  optimal partially-adaptive strategy.
\end{theorem}

\begin{proof}[Proof of Lemma~\ref{lem:scen_aware_generalk_UB}]
  The LP formulation we use for this problem is a variant of
  \eqref{lp-spa} from Section~\ref{sec:model}, with the following
  changes; we introduce variable $y_{st}$ which denotes the extent to
  which scenario $s$ is covered until time $t$ and
  constraint~\eqref{eq:mLP_PA_one_value_per_scen} is replaced by
  constraint~\eqref{eq:pavspaKLP_cover}. 
  For the LP to reflect this additional
  cost we modify constraint~\eqref{eq:mLP_PA_1_box_at_a_time} to~\eqref{eq:timeKLP} so that a box now
  is probed for $p_i$ steps.
  The program \eqref{eq:pavspaKLP_obj}
  is presented below. Denote by $\opt_{t,s}$ and $\opt_{c,s}$ the
  contribution of the query time and cost of scenario $s$ in optimal
  fractional solution. $\alg_{t,s}$ and $\alg_{c,s}$ denote the
  corresponding quantities for the algorithm. 

\begin{alignat}{3}
	\text{minimize}  &\quad \frac{1}{|\scenario|}\sum_{s\in \scenario, t\in \step}(1-y_{st}) & 
	+ &\quad \frac{1}{|\scenario|} \sum_{i\in \boxes ,s\in \scenario,t\in \step} c_{is}z_{ist} &\tag{LP-$k$-cover}\label{eq:pavspaKLP_obj} \\
	\text{subject to} & \hspace{2cm} \sum_{i\in \boxes}x_{it} &  
	=   & \quad1, & \forall t\in \timeset \label{eq:timeKLP} \\
			& \hspace{2cm} \sum_{t\in \step}x_{it} & \leq & \quad 1, & \forall i\in \boxes \notag \\
	  	& \hspace{2.5cm} z_{ist}   & \leq &\quad x_{it}, & \forall s\in \scenario,i\in \boxes,t\in\step \notag \\
   		& \hspace{1.2cm} \sum_{t'\leq t, i\not\in A}z_{ist'} & \geq &\quad (k-|A|)y_{st}, 
	   		& \quad \quad \forall  A \subseteq \boxes, s\in \scenario,t\in \step \label{eq:pavspaKLP_cover}\\
		& \hspace{2.5cm} x_{it}, & z_{ist}, & y_{st} \in [0,1]  
				& \forall s\in \scenario,i \in \boxes, t\in \step \notag
\end{alignat}

The LP formulation we use is exponential in size but we can efficiently find a
separation oracle, as observed in Section~3.1 in \cite{BansGuptRavi2010}. 

	We claim that Algorithm~\ref{alg:PAvsPA_generalK} satisfies the lemma. The
	algorithm first finds an opening sequence by probing boxes with some
	probability at every step, and then select every opened box with some
	probability until $k$ boxes are probed. Note that the number of boxes
	probed at each ``step'' may be more than one. In the algorithm, we set
	constant $\kconst=8$.\\

	\begin{algorithm}[H]
		\caption{ SPA vs PA, k-coverage}\label{alg:PAvsPA_generalK}
		\KwData{Solution $\bm{x}, \bm{y}, \bm{z}$ to above LP, scenario $s$}
		$\sigma:= $  For each phase $\ell=1,2,\ldots$,  
			open each box $i$ independently with probability $q_{i\ell} = \min \lp(\kconst \sum_{t\le 2^\ell} x_{it},1\rp)$.\\
		
                $\tau_s:=$\\

		\quad Define $t_s^*= \max\{t: y_{st}\leq 1/2\}$.\\
	
		\quad \If{$2^\ell\geq t^*_s$}{
			For each opened box $i$, select it with probability $\min\left(\frac{\kconst\sum_{t\le 2^\ell} z_{ist}}{q_{i\ell}},1\right)$.\\
			Stop when we have selected $k$ boxes in total.
		 }
	\end{algorithm}

	\mbox{}\\
	
	Let $t^*_s$ be the latest time at which $y_{st}\leq 1/2$ as in the description of
	the algorithm. As observed in \cite{BansGuptRavi2010} for scenario $s$, we pay
	at least $1-y_{st^{*}_s}\geq\frac{1}{2}$ for each time $t\in [1,t_s^*]$, thus
	\begin{equation}\label{eq:opt_ts_kcoverage}
		\opt_{t,s} \ge \frac{t_s^*}{2}.
	\end{equation}
  
	Fix a scenario $s$. We first analyze the expected probing time
        of the algorithm for this scenario. Denote by
        $\ell_0=\lceil\log t_s^*\rceil$ the first phase during which
        we have a non-zero probability of selecting a box for scenario
        $s$. Notice that for each box $i$, the probability that it is
        selected in phase $\ell\ge\ell_0$ is
        $\min(1,8\sum_{t'\leq 2^{\ell}}z_{ist'})$. The following lemma
        from \cite{BansGuptRavi2010} bounds the probability that in
        each phase $\ell$ such that $2^{\ell}\geq t_s^*$, at least $k$
        boxes are selected.

	\begin{lemma}[Lemma 5.1 in \cite{BansGuptRavi2010}]\label{lem:successprobk}
	If each box $i$ is selected w.p. at least $\min(1,8\sum_{t'\leq t}z_{ist'})$
	for $t\geq t_s^*$, then with probability at least $1-e^{-9/8}$, at least $k$ different boxes are
	selected.	
	\end{lemma}

	Let $\gamma=e^{-9/8}$. Observe that the number of boxes probed
        in a phase is independent of the event that the algorithm
        reaches that phase prior to covering scenario $s$, therefore
        we get

	\begin{alignat}{3}
			\E{}{\text{query time after phase } \ell_0} 
			& = \sum_{\ell=\ell_0}^\infty
                        \E{}{\text{query time in phase } \ell} \cdot  \Pr{\alg\text{ reaches phase } \ell} \nonumber\\
			& \leq  \sum_{\ell = \ell_0}^\infty \sum_{i\in\boxes}\kconst\sum_{t'\leq 2^{\ell}}x_{it'} \cdot
							\prod_{j=\ell_0}^{\ell-1}\Pr{\text{$\leq k$ boxes selected in phase }j} \label{eqn:kquerytime}\\
			& \leq \sum_{\ell = \ell_0}^\infty  2^{\ell}\kconst\cdot \gamma^{\ell-\ell_0} \nonumber\\
			& =
                        \frac{2^{\ell_0}\kconst}{1-2\gamma}<\frac{2t_s^*\kconst}{1-2\gamma}\leq
                        \frac{4\kconst\opt_{t,s}}{1-2\gamma}. \nonumber 
	\end{alignat}

	The second line follows by noting that the algorithm can reach phase $\ell$
	only if in each previous phase there are less than $k$ boxes selected. The
	third line is by Lemma \ref{lem:successprobk} and constraint~\eqref{eq:timeKLP}. 
	 The last line is by $\ell_0=\lceil\log
	t_s^*\rceil$ and inequality (\ref{eq:opt_ts_kcoverage}).  Since the
	expected query time at each phase $\ell$ is at most $\kconst2^{\ell}$, thus the
	expected query time before phase $\ell_0$ is at most
	$\sum_{\ell<\ell_0}\kconst2^\ell<2^{\ell_0}\kconst<2t_s^*\kconst\leq 4\kconst\opt_{t,s}$. Therefore
	the total query time of the algorithm for scenario $s$ is
	\[
		\alg_{t,s}\leq 4\kconst\opt_{t,s}+\frac{4\kconst\opt_{t,s}}{1-2\gamma}<123.25\opt_{t,s}.
	\]

	To bound the cost of our algorithm, we find the expected total
	value of any phase $\ell$, conditioned on selecting at least $k$ distinct boxes 
	in this phase.
	\begin{align*}
			\textbf{E} [  \text{cost in phase }\ell &|  \text{at least $k$ boxes are selected in phase }\ell]  \\
			& \leq \frac{\E{}{\text{cost in phase }\ell}}{\Pr{\text{at least $k$ boxes are selected in phase }\ell}}\\
			& \leq \frac{1}{1-\gamma}\E{}{\text{cost in phase }\ell}\\
			& \leq \frac{1}{1-\gamma}\sum_{i\in\boxes}\kconst\sum_{t\leq 2^{\ell}}z_{ist}c_{is}
			=\frac{1}{1-\gamma}\kconst\opt_{c,s}<11.85\opt_{c,s}.
	\end{align*}

	Here the third line is by Lemma \ref{lem:successprobk} and the
        last line is by definition of $\opt_{c,s}$. Notice that the
        upper bound does not depend on the phase $\ell$, so the same
        upper bound holds for $\alg_{c,s}$. Thus the total cost
        contributed from scenario $s$ in our algorithm is
	\[\alg_s=\alg_{t,s}+\alg_{c,s}<123.25\opt_{t,s}+11.85\opt_{c,s}\leq123.25\opt_{s}.\]
	Taking the expectation over all scenarios $s$, we conclude that the
	scenario-aware strategy gives constant competitive ratio to the
	optimal partially-adaptive strategy.
\end{proof}

\subsection{Picking a matroid basis of rank \texorpdfstring{$k$}{}}
\label{subsec:matroids}

In this section $\mathcal{F}$ requires us to select a basis of a given
matroid.  More specifically, assuming that boxes have an underlying matroid
structure we seek to find a base of size $k$ with the minimum cost and the minimum query time. We first
design a scenario-aware partially-adaptive strategy in
Lemma~\ref{lem:scen_aware_matroid_UB} that is $O(\log k)$-competitive against
optimal partially-adaptive strategy. Then, in Theorem~\ref{thm:matroid_lb} we
argue that such competitive ratio is asymptotically tight.

\begin{lemma}\label{lem:scen_aware_matroid_UB}
	There exists a scenario-aware partially-adaptive $O(\log k)$-approximate
	algorithm to the optimal partially-adaptive algorithm for picking a matroid
	basis of rank $k$.
 \end{lemma}

 Combining this lemma with Theorem~\ref{thm:PA-to-SPA-reduction} we get the
 following theorem.

 \begin{theorem}\label{thm:matroid_ub}
   We can efficiently find a partially-adaptive strategy for optimal
   search over a matroid of rank $k$ that is $O(\log k)$-competitive
   against the optimal partially-adaptive strategy.
 \end{theorem}

The LP formulation is similar to the one for the $k$-coverage constraint,
presented in the previous section.  Let $r(A)$ for any set $A\subseteq \boxes$
denote the rank of this set. The constraints are the same except for
constraints \eqref{eq:LP_mat_less_than_rank} and
\eqref{eq:LP_mat_subsets_rank} that ensure we select no more than the rank of
a set and that the elements that remain unselected are adequate for us to cover
the remaining rank respectively.
%
	\begin{alignat}{3}
		\text{minimize}   \quad & \frac{1}{|\scenario|}\sum_{s\in \scenario, t\in \step}(1 - y_{st}) &\quad  + \quad &\frac{1}{|\scenario|}\sum_{i\in \boxes,s \in \scenario, t\in \step}c_{si}z_{ist}  \tag{LP-matroid} \\
		\text{subject to} \quad  &\hspace{1.4cm} \sum_{i\in \boxes}x_{it}  & \quad =\quad & 1,      & \forall t\in \timeset\notag \\
			   & \hspace{1.4cm}\sum_{t\in \step}x_{it}          & \quad \leq \quad  & 1,   & \forall i\in \boxes\notag \\ 
			   & \hspace{0.7cm} \sum_{t\in \step, i\in
                             A}z_{ist} & \quad \leq	\quad &  r(A),
                           & \forall s\in \scenario, A\subseteq
                           \boxes \label{eq:LP_mat_less_than_rank}\\ 
			   & \hspace{2cm} z_{ist} 	  & \quad \leq \quad & x_{it}, & \forall	 s\in \scenario  ,i\in \boxes ,t\in \step \notag\\ 
			   &\hspace{0.65cm} \sum_{i\not\in A}\sum_{t'\leq t}z_{ist'} &	 \quad \geq \quad & (r([n])-r(A))y_{st},\quad 
				   & \forall A\subseteq \boxes, s\in
                                   \scenario ,t\in
                                   \step \label{eq:LP_mat_subsets_rank}\\ 
		& \hspace{2.5cm} x_{it}, & z_{ist}, & y_{st} \in [0,1]  
				& \forall s\in \scenario,i \in \boxes, t\in \step \notag
	\end{alignat}

\paragraph{Solving the LP efficiently}
The LP formulation we use is exponential in size but we can efficiently find a
separation oracle.  Every set of constraints can be verified in polynomial time
except for constraints (\ref{eq:LP_mat_subsets_rank}).  Rewriting these last
constraints we get

\[
	\sum_{i}\sum_{t' \leq t}z_{ist'} -\sum_{i \in A}\sum_{t' \leq t}z_{ist'}  
		\geq r([n]) -r(A), \;\;\; \forall A\subseteq \boxes, t\in\step.
\]

Then the problem is equivalent to minimizing the function $ g(A) = r(A) -
\sum_{i \in A}\sum_{t' \leq t}z_{ist'} $ over all subsets of items $A \subseteq \boxes$. The
function $g(A)$ is submodular since the rank function $r(A)$ is submodular,
therefore we can minimize it in polynomial time \cite{GrotLovaSchr1981}. The
formal statement of the main theorem is the following.

\begin{proof}[Proof of Lemma~\ref{lem:scen_aware_matroid_UB}]
	We claim that Algorithm~\ref{alg:PAvsPA_matroid} satisfies the lemma. The
	algorithm first finds an opening sequence by probing boxes with some
	probability at every step, and then knowing the scenario selects every
	opened box with some probability until a basis of rank $k$ is found. In 
	the algorithm we set constant $\mconst=64$.\\

	\begin{algorithm}[H]
		\caption{ SPA vs PA, matroid}\label{alg:PAvsPA_matroid}
		\KwData{Solution $\bm{x}, \bm{y}, \bm{z}$ to above LP, scenario $s$}
	
		$\sigma :=$ for every $t =1,\ldots, n$, open each box $i$ independently
			with probability $q_{it} = \min \lp\{\mconst \ln k\frac{\sum_{t'\le t}
			x_{it'} }{t},1\rp\}$. \\
		
                $\tau_s:=$\\
                \quad Let $t_s^*= \min\{t: y_{st}\leq 1/2\}$.\\
		\quad \If{$t>t^*_s$}{
			For each opened box $i$, select it with probability $\min \lp\{\frac{\mconst\ln k\sum_{t'\le t} z_{ist'}}{t q_{it}} ,\ 1\rp\}$.\\
			Stop when we find a base of the matroid.
		}
	\end{algorithm}

	\mbox{}\\

	In scenario $s$, let phase $\ell$ be when $t\in (2^{\ell-1}t_s^*, 2^\ell t_s^*]$.
	The proof has almost identical flow as the proof for $k$-coverage case. 
	We still divide the time after $t^*_s$ into exponentially increasing
	phases, while in each phase we prove that our success probability is a constant. 
	The following lemma gives an upper bound for the query time needed 
	in each phase to get a full rank base of the matroid. The proof is deferred
	to Section~\ref{sec:apndx_extension} of the appendix.

	\begin{lemma}\label{lem:rank}
		In phase $\ell$, the expected number of steps needed to select a set of full rank is at most
		$(4+2^{\ell+2} /\mconst)t^*_s$.
	\end{lemma}

	Define $\mathcal{X}$ to be the random variable indicating number of steps needed 
	to build a full rank subset. The probability that we build a full rank basis 
	within some phase $\ell\geq 6$ is
	\begin{align}\label{eq:prob_cont}
		\Pr{\mathcal{X} \leq 2^{\ell-1}t_s^*} 
				 \ge  1 - \frac{\E{}{\mathcal{X}}}{2^{\ell-1}t_s^*} 
				 \ge  1 - \frac{1}{2^{\ell-1}t_s^*}(4+2^{\ell+2} /\mconst)t^*_s 
				=1-2^{3-\ell}-\frac{8}{\mconst}\geq \frac{3}{4},
	\end{align}
	where we used Markov's inequality for the first inequality and
	Lemma~\ref{lem:rank} for the second inequality.  To calculate the total
	query time, we sum up the contribution of all phases. 
	\begin{align}
			\E{}{\text{query time after phase 6}} &= 
					\sum_{\ell=6}^\infty \E{}{\text{query time at phase } \ell} \cdot  \Pr{\alg\text{ reaches phase } \ell}  \nonumber\\
			 & \leq \sum_{\ell = 6}^\infty  \sum_{t=2^{\ell-1}t_s^*+1}^{2^{\ell}t_s^*}\sum_{i\in\boxes}\mconst\ln k\cdot\frac{\sum_{t'\leq t}x_{it'}}{t}\lp(\frac{1}{4}\rp)^{\ell-6} \label{eqn:matroidquerytime}\\
			 & \leq \sum_{\ell = 6}^\infty  2^{\ell-1}t_s^*\mconst\ln k\cdot \lp(\frac{1}{4}\rp)^{\ell-6}  \nonumber\\
			 & = \frac{128\mconst\ln kt_s^*}{3}\leq\frac{256c\ln k\opt_{t,s}}{3} \nonumber. 
	\end{align}

	Here the second line uses that each box $i$ is
        probed at each time step $t$ with probability $\mconst\ln k\cdot\frac{\sum_{t'\leq t}x_{it'}}{t}$. 
        The third line follows from constraint~\eqref{eq:timeKLP}. The last line
        uses $t_s^*\leq 2\opt_{t,s}$ by \eqref{eq:opt_ts_kcoverage}. Since the expected query time at
        each step is $\mconst\ln k$ and there are $2^5t_s^*\leq 64\opt_{t,s}$
        steps before phase $6$, we have
	\[
		\alg_{t,s}\leq \mconst\ln k\cdot 64\opt_{t,s}+\frac{256\mconst\ln k\opt_{t,s}}{3} 
			= O(\log k)\opt_{t,s}.
	\]

	As for $k$-coverage case, to bound the cost of our algorithm,
        we find the expected total cost of any phase $\ell\geq 6$,
        conditioned on boxes forming a full rank base are selected in
        this phase.
	\begin{align*}
			\textbf{E} [\text{cost in phase }\ell| & \text{full rank base selected in phase }\ell]\\
		 & \leq \frac{\E{}{\text{cost in phase }\ell}}{\Pr{\text{full rank base selected in phase }\ell}}\\
		&\leq \frac{1}{3/4}\E{}{\text{cost in phase }\ell}\\
		&\leq \frac{1}{3/4}\sum_{i\in\boxes}\sum_{t=2^{\ell-1}t_s^*+1}^{2^{\ell}t_s^*}\mconst\ln k
				\frac{\sum_{t'\leq t}z_{ist'}c_{is}}{t}\\
		&\leq \frac{1}{3/4}\sum_{t=2^{\ell-1}t_s^*+1}^{2^{\ell}t_s^*}\mconst\ln k\sum_{i\in\boxes}
				\frac{\sum_{t'\in\step}z_{ist'}c_{is}}{2^{\ell-1}t_s^*}\\
		&=\frac{1}{3/4}\mconst \ln k\opt_{c,s}
			=O(\log k)\opt_{c,s}.
	\end{align*}
	Such upper bound of conditional expectation does not depend on $\ell$, thus
	also gives the same upper bound for $\alg_{c,s}$. Therefore
	$\alg_s=\alg_{t,s}+\alg_{c,s}\leq O(\log k)(\opt_{t,s}+\opt_{c,s})=O(\log
	k)\opt_{s}$.  Take expectation over $s$, we have the scenario-aware
	adaptive strategy Algorithm \ref{alg:PAvsPA_matroid} is $O(\log
	k)$-competitive against the optimal partially-adaptive strategy.
\end{proof}

Now we argue that the $O(\log k)$-approximation we got is essentially tight. 
The following theorem implies that under common complexity assumption,
no efficient fully-adaptive algorithm can get asymptotically better competitive ratio,
even compared to optimal non-adaptive cost.

\begin{theorem}\label{thm:matroid_lb}
  Assuming NP$\not\subseteq$RP, no
  computationally efficient fully-adaptive algorithm 
  can approximate the optimal non-adaptive cost within a factor of $o(\log k)$.
\end{theorem}

\begin{proof}
	We provide an approximation-preserving reduction from Set Cover problem 
	to finding good fully-adaptive strategy.
	Let $\mathcal{SC} = \lp( [n], \{S_1, \ldots, S_k\}\rp)$ be a Set Cover
	instance on a ground set of $n$ elements, and $k$ sets $S_1,\ldots, S_k$.
	Denote by $\opt_{SC}$ the optimal solution to this Set Cover instance.  We
	construct an instance of partition matroid coverage, where the rank is $k$.
	Each segment of the partition consists of multiple copies of the sets
	$S_1,\ldots, S_k$.  Every scenario consists of one set from each segment,
	as seen in Table~\ref{table:reduction}.

	\begin{table}[H]
			\centering
		\begin{tabular}{lcccc}
		\multicolumn{1}{l}{}    & \multicolumn{1}{l}{\textbf{Segment $1$}} &
		\multicolumn{1}{l}{\textbf{Segment $2$}} & \multicolumn{1}{l}{$\ldots$}
												 &
												 \multicolumn{1}{l}{\textbf{Segment
												 $k$}} \\
		\textbf{Scenario $1$}   & $S_1$    & $S_1$     & $\ldots$  & $S_1$  \\
		\textbf{Scenario $2$}   & $S_1$    & $S_1$     &   & $S_2$  \\
		$\ldots$                & $\ldots$ & $\ldots$  &  $\ldots$ & $\ldots$  \\
		\textbf{Scenario $k^k$} & $S_k$    & $S_k$     &   & $S_k$
		\end{tabular}
			\caption{Instance of partition matroid $k$-coverage}
				\label{table:reduction}
	\end{table}

	A scenario is covered when $k$ elements are selected, one for each of the $k$ sets of
	every segment. This is an instance of the probing problem we study
	with cost for each box being 0 or $\infty$. 
	Similarly, we say a segment is \emph{covered} when we have
	chosen at least one element in every set it contains. Denote by $\alg_{FA}$
	and $\opt_{NA}$ the solution of any fully-adaptive algorithm and the
	optimal non-adaptive solution respectively for this transformed instance.
	Any fully-adaptive algorithm will select elements, trying to cover all
	scenarios. Initially, observe that $\opt_{NA}\leq k \opt_{SC}$, since the
	non-adaptive will at most solve the Set Cover problem in the $k$ different
	segments. We assume that we can approximate the non-adaptive strategy with
	competitive ratio $\alpha$ i.e.  $\alg_{FA} = O(\alpha)\opt_{NA} = O(k\alpha) \opt_{SC}$.\\

	Let $\ell$ be the number of elements $\alg_{FA}$ has selected when exactly
	$k/2$ segments are covered and let $s$ be a randomly chosen
	scenario.  For each one of the $k/2$ uncovered segments, there are at least
	1 uncovered set. Therefore
	\[
		\Pr{s \text{ is uncovered} } \geq 1-\lp(1-\frac{1}{k} \rp)^{k/2} \approx 1-\frac{1}{\sqrt{e}}.
	\]

	This implies $\alg_{FA} \geq \ell \lp( 1-\frac{1}{\sqrt{e}} \rp)$ thus
	$\ell = O(k \alpha)\opt_{SC}$. Notice that there exists some segment that is
	covered using $\ell/(k/2) = O(\alpha)\opt_{SC}$ elements. Thus any efficient algorithm 
	that provides $O(\alpha)$-approximation of non-adaptive strategy using fully-adaptive strategy
	can be transformed efficiently to an $O(\alpha)$-approximation algorithm for Set Cover.

	Although above reduction from set cover has $k^k$ scenarios that cannot be
	constructed in polynomial time, by Lemma~\ref{lem:learn}
	poly$(n,\frac{1}{\epsilon},\log\frac{1}{\delta})$
	samples of all scenarios is sufficient to get accuracy within $\epsilon$ with
	probability $1-\delta$ for any probing strategy. 
	Let $\e=1$ and $\delta=\frac{1}{3}$. The above reduction implies that if
	there is a poly-time algorithm that computes a probing strategy with 
	cost $o(\log k)\opt_{NA}$, there exists a poly-time algorithm to solve 
	Set Cover with competitive ratio $o(\log k)$ with
	probability $\frac{2}{3}$. By \cite{DinuSteu2014} such algorithm cannot exist
	assuming NP$\not\subseteq$RP.
\end{proof}

\section{Boxes with General Probing Times: Revisiting the Main Results}\label{sec:generalprobingcost}

In this section, we consider settings where different boxes require
different amounts of time to probe. Let $p_i$ denote the probing time
required to probe box $i$.  We assume $p_i\in [1,P]$ for some $P$ that
is polynomially large in $n$. The running time and sample complexity
of our algorithms will depend linearly on $P$. Henceforth we will
assume that the $p_i$'s are integers: rounding up each probing time to
the next integer only increases the total objective function value by
a factor of at most $2$.

\subsection{Ski rental with general rent cost and learnability via sampling}
We first investigate the learnability of the optimal search algorithm via
sampling polynomially many scenarios, i.e.
Theorem~\ref{thm:PA-to-SPA-reduction}. Recall that Theorem~\ref{thm:PA-to-SPA-reduction}
requires two building blocks: Corollary~\ref{lem:scenario_aware_approx} which
describes a reduction from a general strategy to a scenario-aware strategy and
Lemma~\ref{lem:learn} that guarantees that a small sample over scenarios suffices to achieve a
good approximation.

We proved Lemma~\ref{lem:learn} by observing that for
any probing order $\pi$, the cost of any scenario $s$ is bounded in a
polynomial range.  This still holds since the total probing time is
bounded by $nP$.

In order to show Corollary~\ref{lem:scenario_aware_approx} in our case, we need to
solve a further generalization of the ski rental problem where we have arbitrary rent costs.
Specifically, in the ski rental problem with general rent cost, the input is a sequence
of non-increasing buy costs, $a_1\ge a_2\ge a_3 \ge \ldots$ as well as an integral rent
costs $p_t$ for each time $t$. At each step $t$, the algorithm decides to either rent skis at a cost
of $p_t$, or buy skis at a cost of $a_t$. We show that Lemma~\ref{lem:gen_ski_rental} 
still holds beyond the unit-rental-cost case. Together with Lemma~\ref{lem:learn} we recover Theorem~\ref{thm:PA-to-SPA-reduction}.
\begin{lemma}\label{cor:varying_rent}
	 Consider any sequence of integral buy cost $a_1\ge a_2\ge \ldots$ and integral rent cost $p_1,p_2,\cdots$. 
   There exists an online algorithm that
	chooses a stopping time $t$ so that
	\[ \sum_{i=1}^{t-1}p_i + a_{t} \le \frac{e}{e-1} \min_{j} \lp\{ \sum_{i=1}^{j-1}p_i +a_j\rp\}.\]
\end{lemma}
\begin{proof}[Proof of Lemma~\ref{cor:varying_rent}]
	The case with the general rental cost is equivalent to the following
	unit-rent-cost problem: at time $t=\sum_{i=1}^{j-1}p_i$, the buyer can
	decide to either pay $a_j$ for buying skis, or continue to rent for $p_j$
	consecutive time slots, each with rent cost 1, and then get to see the next
	possible skis buying cost $a_{j+1}$.  The two problems have the same
	offline optimal cost.
	
	To solve the case with general rental cost, we use the algorithm in
	Lemma~\ref{lem:gen_ski_rental} as a subroutine.  Assume that in the
	general-rental-cost case, we have already rented skis for $j-1$ days, the
	total rental cost we have paid is $t\sum_{i=1}^{j-1}p_i$. Now we see the
	next buy and rent values $a_j$ and $p_j$.  To decide what to do at the
	current step, we run the unit-rental-cost algorithm as in the proof of
	Lemma~\ref{lem:gen_ski_rental} for additional $p_j$ time steps without
	doing real probing, because we know the buying cost will not change in the
	next $p_j$ unit time steps.  If the algorithm with unit rental cost does
	not stop in the following $p_j$ unit time steps in the simulation, we
	decide to do the same, i.e. paying the rental cost $p_j$ at the current
	time step. If the algorithm with unit rental cost stops and buys skis in
	the following $p_j$ unit time steps in the simulation, we decide to buy
	skis immediately, which results in a better total cost than in the
	corresponding unit-rental-cost case.  Since the algorithm in the previous
	lemma pays an $\frac{e}{e-1}$-approx to the optimal offline cost of the
	corresponding unit-rental-cost case, our algorithm for the
	general-rental-case is no worse than it, thus an $\frac{e}{e-1}$-approx to
	the optimal offline cost of the general-rental-cost case.
	
  \end{proof}

\subsection{Linear program formulations}\label{sec:lp-general}
To get the linear program relaxation of the optimal Non-Adaptive strategy for
selecting one box, we only need to change the objective function of the linear
program.
\begin{alignat}{3}
  \text{minimize}  \quad & \sum_{i\in \boxes} x_i p_i   & \quad  + \quad
  & \frac{1}{|\scenario|}\sum_{i\in \boxes,s\in \scenario}c_{is}z_{is}
  & \tag{LP-NA-General} \label{lp-na-general}\\
  \text{subject to}\quad & \sum_{i\in \boxes}z_{is} & \quad  = \quad   & 1, 				& \forall s\in \scenario \notag\\
						 & \hspace{0.7cm} z_{is} 			 & \quad \leq \quad & x_i , & \forall i\in \boxes, s\in \scenario \notag\\
						 &  \hspace{0.2cm} x_i, z_{is} 		 &\quad \in \quad & [0,1] & \forall i\in \boxes, s\in \scenario \notag
\end{alignat}
For the LP of optimal SPA strategy for selecting one box, we need to account
for the probing time of every box in the constraint. In order to do that, we
will require that every box is being probed for $p_i$ consecutive steps:
$x_{it}=1$ means that box $i$ has been probed since time $t-p_i+1$, and the
probing of the box finishes at time $t$. Thus at each time step $t$, there are
$\sum_{i\in\boxes}\sum_{t\leq t'\leq t+p_i-1}x_{it'}$ boxes under probing, and
this should be upper bounded by 1. The rest of the program will be the same.
Since the probing time of each box is polynomially bounded, such LP still has a
polynomial size.
  \begin{align}
    & \text{minimize} & \frac{1}{|\scenario|} \sum_{i\in\boxes, s\in
    \scenario, t\in \timeset}tz_{ist} & + \frac{1}{|\scenario|}\sum_{i\in
                                   \boxes ,s\in \scenario,t\in \timeset}
                                   c_{is}z_{ist} &  \tag{LP-SPA-General} \label{lp-spa-general}\\
    & \text{subject to} & \sum_{i\in \boxes}\sum_{t \leq t' \leq t+p_i-1}x_{it'} & \leq 1, & \forall t\in \timeset \label{eq:LP_PA_General_change}\\
    & & \sum_{t\in \timeset}x_{it}  & \leq 1, & \forall i\in \boxes \notag \\
    & & z_{ist} & \leq x_{it}, & \forall s\in \scenario,i\in \boxes,t\in\timeset \notag\\
    & & \sum_{t'\in\timeset, i\in \boxes}z_{ist'} & = 1, & \forall s\in \scenario \notag\\
    & & x_{it}, z_{ist} & \in [0,1]  & \forall s\in \scenario,i \in \boxes, t\in \timeset \notag
  \end{align}
For the case of selecting $k$ boxes or picking a matroid basis of rank $k$, the
change to the LP of SPA strategy would be the same: replacing the first
constraint ``$\sum_{i\in\boxes}x_{it}=1$'' by \eqref{eq:LP_PA_General_change}.

\subsection{SPA vs NA: selecting a single item}
We show that Algorithm~\ref{alg:PAvsNA_1_cov} works for the
general-probing-time case with approximation ratio only losing a factor of 2. 
\begin{lemma}\label{lem:pavsna_exists_sequence_general}
	In general-probing-times case, we can efficiently compute a scenario-aware
	partially-adaptive strategy with competitive ratio $2$ against the optimal
	non-adaptive strategy. 
\end{lemma}
\begin{proof}
  The analysis of $\alg_c$ remains the same,
  i.e. $\E{}{\alg_{c}}\leq\opt_{c}$.  Now we consider $\alg_t$. Notice
  that each step of the algorithm for constructing the probing order
  is completely independent, with stopping probability
  $\frac{1}{\sum_{i\in\boxes}x_i}$ at each point. However, the
  ``length'' of each step depends on the probing time for the box
  picked for that step. Let $\tau$ denote the step at which we
  stop. We have $\E{}{\tau}=\sum_{i\in\boxes}x_i$. For any step
  $t<\tau$, the expected probing time for this step is
\[
	\E{}{\textrm{probing time at step }t|t\textrm{ is not stopping time}}<\frac{\E{}{\textrm{probing time at step }t}}{\Pr{t\textrm{ is not stopping time}}}=\frac{\sum_{i\in\boxes}x_ip_i/\sum_{i\in\boxes}x_i}{1-1/\sum_{i\in\boxes}x_i}.
	\]
Thus the expected total probing time is
\begin{eqnarray*}
\E{}{\alg_t}&=&\E{}{\textrm{probing time at steps }<\tau}+\E{}{\textrm{probing time at step }\tau}\\
&\leq& \frac{\sum_{i\in\boxes}x_ip_i/\sum_{i\in\boxes}x_i}{1-1/\sum_{i\in\boxes}x_i}(\E{}{\tau}-1)+\sum_{i\in\boxes}x_ip_i\\
&=&2\sum_{i\in\boxes}x_ip_i=2\opt_t.
\end{eqnarray*}
Thus $\E{}{\alg}\leq 2\opt$.
\end{proof}

\subsection{SPA vs PA: \texorpdfstring{$k$}{}-coverage and matroid base}
Now we show the algorithms for the case of selecting $k$ boxes and selecting a
matroid base still works when we have general probing times. The only
difference is that the algorithms will now base on the modified LP in
Section~\ref{sec:lp-general}.

In the entire analysis of the two cases, the only place where we employ probing
times of boxes is when we try to bound the expected total probing time of each
phase $\ell$ in \eqref{eqn:kquerytime} and \eqref{eqn:matroidquerytime}
respectively.  These expected probing time terms, $\sum_{t'\leq 2^\ell}x_{it'}$
in \eqref{eqn:kquerytime} and $\sum_{t'\leq t}x_{it'}$ in
\eqref{eqn:matroidquerytime}, will get changed to $\sum_{t'\leq
2^\ell}p_ix_{it'}$ and $\sum_{t'\leq t}p_ix_{it'}$ respectively.

Now we argue that the proof will still go through step by step, and it suffices
to show that $\sum_{i\in\boxes}\sum_{t'\leq t}p_ix_{it'}\leq t$. Sum up LP
constraint \eqref{eq:LP_PA_General_change} from $1$ to $t$, we have
\[
	\sum_{i\in \boxes}\sum_{t'\leq t}\sum_{t' \leq t'' \leq t'+p_i-1}x_{it''}\leq t.
\]

Notice that for any $t'\leq t$, $x_{it'}$ appears exactly $p_i$ times in the
sum. The counting argument implies
\[	
	\sum_{i\in\boxes}\sum_{t'\leq t}p_ix_{it'}\leq \sum_{i\in \boxes}\sum_{t'\leq t}\sum_{t' \leq t'' \leq t'+p_i-1}x_{it''}\leq t.
	\]

	Observe that for the case of $k=1$ discussed in Section~\ref{sec:pavspa}
	this extension implies the $124$-approximation of
	Theorem~\ref{thm:kcoverage}. We believe that the argument can be tightened
	to obtain a much better factor for $k=1$ but do not attempt to optimize the
	constant. We also note that our reduction to MSSC in
	Section~\ref{sec:pavspa} continues to work with general probing times,
	however this general setting has not been studied previously for MSSC.

\section{Inapproximability of the profit maximization variant}
\label{sec:max}

In this section we consider the profit maximization variant of the problem
discussed above. The boxes now contain some prize value $v_{is}$ for each box $i$
in scenario $s$, and we want to
maximize expected \emph{profit}. Formally, let $\mathcal{P}_s$ be the set of
probed boxes in scenario $s$, our objective is to maximize

\[
	\E{s}{ \max_{i\in \mathcal{P}_s} v_{is} - |\mathcal{P}_s|}.
\]

It turns out that, contrary to the minimization case, obtaining a constant
approximation in this setting is impossible, as the following theorem shows.

\begin{theorem}
  Assuming P$\neq$NP, no computationally efficient fully-adaptive algorithm can
  approximate the optimal non-adaptive profit within a constant factor.
\label{thm:max_LB}
\end{theorem}

The proof follows similarly to the minimization case, where we use again
Lemma~\ref{lem:LB_max_coverage} to construct a bad instance that can give
arbitrarily bad approximation.

\begin{proof}[Proof of Theorem~\ref{thm:max_LB}]
	Let $H>0$ and $p\in [0,1]$ be appropriate constants, to be determined
	later. Let $\mathcal{SC} = ([m], \{S_1, \ldots, S_n\})$ be a set cover
	instance with $m$ elements and $n$ sets. Denote its optimal value by
	$\opt_{SC}$.  To transform this into an instance of the search problem,
	every element $e_j \in [m]$ corresponds to a scenario $j$, and every set
	$S_i$ to a box $i$. We set $v_{ij}=H$ iff $e_j \in S_i$, otherwise $v_{ij}
	= 0$. We also add a new scenario $X$ with $v_{Xi} = 0,\ \forall i\in [n]$.
	Scenario $X$ occurs with probability $p$ and all the other $m$ scenarios
	happen with probability $(1-p)/m$ each.  Observe that contrary to the
	minimization lower bound of Section~\ref{sec:favsna}, the additional
	scenario $X$ has a low value ($0$ instead of $H$), and the other scenarios
	give a high value ($H$) when covered. 

	In this instance, the profit of optimal non-adaptive strategy is $\opt_{NA}
	\geq (1-p)H - \opt_{SC}$, since we may pay the set-cover cost to find a box
	with value $H$ in every scenario other than $X$.

	Now let us consider any computationally efficient algorithm $\al$ that
	returns a fully adaptive strategy for such an instance. Since the values of
	the boxes are $0$ or $H$, we may assume without loss of generality that any
	fully-adaptive strategy stops probing as soon as it observes a box with
	value $H$ and chooses that box. We say that the adaptive strategy covers a
	scenario when it finds a box of value $H$ in that scenario. Observe that
	similarly to the lower bound of subsection~\ref{sec:favsna}, any FA
	strategy will follow a probing order independent of the scenario, which we
	can then convert to an approximate solution for the underlying set cover
	instance.  Then, for any constant $\e>0$, by
	Lemma~\ref{lem:LB_max_coverage}, there must exist a set cover instance and
	correspondingly an instance of the search problem, such that for the
	adaptive strategy returned by the algorithm for that instance, for every
	$k$, the fraction of scenarios other than $X$ covered before step $k$ is at
	most $1-\lp(1-\frac{1+\e}{\opt_{SC}}\rp)^{k-1}$. Consider such an instance
	and let $N$ denote the maximum number of boxes the strategy probes before
	stopping to return a box of value $H$. 
		
	The same as \eqref{eqn:lb_probingcost}, the expected query time of the strategy is at least
	\begin{eqnarray*}
	pN+(1-p)\lp(1-\lp(1-\frac{1+\e}{\opt_{SC}}\rp)^N\rp)\frac{\opt_{SC}}{1+\e}.
	\end{eqnarray*}
	On the other hand, the expected value obtained by the fully-adaptive strategy is at most
	\[H\cdot\Pr{s\neq X\wedge \text{FA finds value $H$ in first $N$ steps}}\leq (1-p)H\lp(1-(1-\frac{1+\e}{\opt_{SC}})^{N}\rp).\]
	Thus the profit of such fully-adaptive strategy is upper bounded by
	\begin{equation*}
			\alg_{FA} \leq (1-p) H \cdot \lp( 1 - \lp( 1 - \frac{1 + \e}{\opt_{SC}} \rp)^N \rp) 
			- (1-p) \lp( 1 - \lp(1-\frac{1+\e}{\opt_{SC}}\rp)^N\rp)\frac{\opt_{SC}}{1+\e} - pN.
	\end{equation*}
	Let $x$ be defined so that 	$( 1 - \frac{1 + \e}{\opt_{SC}})^N =e^{-x}$.
	Then, $N=-x/\ln(1-\frac{1 + \e}{\opt_{SC}})\ge x(\frac{\opt_{SC}}{1 +
	\e}-1)$. Substituting these expressions in the above equation we get
	\begin{equation*}
			\alg_{FA}(x) \leq  (1-e^{-x}) (1-p)\lp( H - \frac{\opt_{SC}}{1+\e} \rp) - px\lp(\frac{\opt_{SC}}{1+\e}-1\rp).
	\end{equation*}

	\noindent
	Observe that the right hand side is maximized at $x= \ln \lp(
	\frac{(1-p)\lp( (1+\e)H - \opt_{SC}\rp) }{p (\opt_{SC}-(1+\e))}\rp) $. By setting $(1-p)H = \frac{2\e+1}{\e+1}\opt_{SC}$, 
	$\e\to 0$ and $p\to 1$, we get
	\begin{eqnarray*}
			\frac{\alg_{FA}}{\opt_{NA}} \leq 2 - \frac{p}{\e}\log \lp( \frac{2\e}{p}+1\rp)\to 0
	\end{eqnarray*}

	when $\opt_{SC}\to\infty$. Thus no efficient fully-adaptive algorithm can
	approximate the optimal non-adaptive profit within any constant factor.
	\end{proof}

\newpage 

\appendix
\section{Proofs from Section~\ref{sec:reduction}}\label{sec:lemma_proof}

\begin{proof}[Proof of Lemma~\ref{lem:gen_ski_rental}]
	We prove that Algorithm \ref{alg:ski_rental} satisfies the lemma.  The
	algorithm sees an instance $\mathcal{I} = \{a_1, a_2, \ldots \}$ and
	essentially starts a new Ski Rental problem every time it finds a lower
	$a_t+t-1$ value, using the $\frac{e}{e-1}$-competitive 
	randomized Ski Rental algorithm \cite{KarlManaMcgeOwic1990} as a black
	box, to choose the new stopping time $\tau$.  

	\begin{algorithm}
		\KwData{ ski($C$): random stopping time according to Ski Rental with buying cost $C$}
		\KwIn{Sequence $a_1, a_2, \ldots $ of buying costs}
		$C = \infty$, $\tau = \infty$\\
		\ForEach{time $t\geq1$}{
			\If{$a_t+t-1 < C$}{
					$C = a_t + t-1$\\ \label{algline:step}
				$\tau = t-1+$ ski$(C-t+1)$\\
			}
			\If{$t = \tau$}{
				Buy at price $\min_{t'\leq t}\{a_{t'}\}$
			}
		}
		\caption{Ski Rental for time-varying buying prices}
		\label{alg:ski_rental}
	\end{algorithm}

	Every time the algorithm changes the current cost value $C$
	(line~\ref{algline:step} of Algorithm~\ref{alg:ski_rental}) we say the
	sequence $\{a_t\}$ takes a \emph{step}. 
	Suppose that the sequence takes a step at time $1=t_N<t_{N-1}<\cdots<t_1$. Notice that
	$N\leq a_1$, since the sum of renting cost and buying cost after time $t=a_1$ will be 
	at least $a_1$, which is the total cost at time 1.
	Denote by $\alg^k$ the algorithm's cost on instance $a_{t_k},a_{t_{k+1}},\cdots$, which
	is the truncated instance that takes $k$ steps until the end. 
	Define $\opt^k = \min_t\{a_t + t - 1\}$ be the optimal cost for the
	same instance that the algorithm takes $k$ steps until the end. We claim that $\E{}{ \alg^k
	} \leq e/(e-1)\opt^k$ for any $k$, and prove the claim using induction on
	the number of steps. 
	
	Observe that for
	$k = 1$, since the algorithm only takes a step at the beginning, we have 
	$a_t + t -1 \geq a_1$ for any $t\geq1$. In this case, $\opt^1=a_1$ 
	and $\alg$ only considers $a_1$ as buying cost. This is
	exactly a special case of the traditional Ski Rental problem with one
	buying cost. Therefore we get $\alg^1\leq e/(e-1)\opt^1$.
	
	\noindent
 	For any $k>1$, denote by $T$ be the first time the algorithm takes a step, $\tau_0$ 
	the first stopping time set by the algorithm. The expected cost of the algorithm is
	\begin{align*}
		\E{}{ \alg^k }  &=  \E{}{ \alg^k \mathbbm{1}_{ \{ \tau_0 \leq T \} } } + \E{}{ (T + \alg^{k-1}) \mathbbm{1}_{ \{ S > T \} } }&  \\
						&\leq \E{}{ (a_1 + \tau_0) \mathbbm{1}_{ \{ \tau_0 \leq T \} } }   + T \Pr {\tau_0 > T} + \E{}{ \alg^{k-1} }&  \\
						&\leq  \E{}{ (a_1 + \tau_0) \mathbbm{1}_{ \{ \tau_0 \leq T \} } }   
					+ T \Pr {\tau_0 > T} + \frac{e}{e-1} \opt^{k-1}& \\
		  &\leq \frac{e}{e-1} T + \frac{e}{e-1} (\opt^{k}-T)&\\
		  &= \frac{e}{e-1} \opt^k.
	\end{align*}
	Here the second line comes from the algorithm's cost when $\tau_0\leq T$ is exactly $a_1+\tau_0$, which is paying the buying cost $a_1$ at time $1$ and the renting cost for $\tau_0$ rounds. The third line is by inductive hypothesis. The fourth line is true since for a ski-rental instance with $T$ days and buying cost $a_1>T$, renting for $T$ days is optimal, while the strategy of using $\tau_0$ as stopping time has cost $\E{}{ (a_1 + \tau_0) \mathbbm{1}_{ \{ \tau_0 \leq T \} } } + T \Pr {\tau_0 > T}$ should give $\frac{e}{e-1}$-approximation to optimal. The last line comes from the fact that $\opt^k=\opt^{k-1}+T$, as the optimal solution always choose to rent in the first $T$ steps.

	By induction, $\alg^N\leq\frac{e}{e-1}\opt^N$.
  \end{proof}





  \begin{proof}[Proof of Corollary~\ref{lem:scenario_aware_approx}]
	Recall that a scenario-aware strategy consists of a sequence and a scenario
	dependent stopping rule. Let  $(\sigma,  \tau)$ be the scenario-aware
	partially-adaptive strategy. For the case with unit query time, 
	by running the algorithm described in the
	proof of Lemma~\ref{lem:gen_ski_rental} using sequence $\sigma$ as
	input\footnote{For this reduction, we set the costs $a_t=1 + \min_{i\leq t}
			(c_{\sigma(i)s} )$ so that we have a decreasing sequence as
	Lemma~\ref{lem:gen_ski_rental} requires.}, we obtain a stopping rule that
	does not depend on the scenario and only worsens the approximation by a
	factor of $e/(e-1)$. Similarly for the case of general probing time, we will use the algorithm
	for Corollary~\ref{cor:varying_rent}.
  \end{proof}

\section{Proofs from Section~\ref{sec:favsna}} \label{sec:apndx_minLB}

\begin{proof}[Proof of Lemma~\ref{lem:LB_max_coverage}]
	Assume that there exists an algorithm $\mathcal{A}$ such that for every
	instance of Set Cover finds $k$ sets that cover at least
	$1-(1-\frac{1+\e}{\opt})^k$ of the elements.  Given an instance of Set
	Cover, with a ground set of $n$ elements, we repeatedly run $\mathcal{A}$
	on the set of uncovered elements left at each round. Every time, we create
	a new instance with ground set composing of only the elements left uncovered in the
	previous round. Using the guarantee for $\mathcal{A}$ in each round $i$ we
	cover at least $(1-(1-\frac{1+\e}{\opt})^{k_i})$ for some $k_i\in\lp[
	1,\log n \frac{\opt}{1+\e} \rp]$.

	Denote by $z$ the number of rounds we need to cover all elements and by
	$k_i$ the number of the elements we cover at round $i$. In the end of round
	$z$ there are 
	\begin{equation}
		n\prod_{i=1}^z \lp(1 - \frac{1+\e}{\opt}\rp)^{k_i} 
	\end{equation}

\noindent
elements left uncovered. When the above quantity equals $1$, we are left with
$1$ element and the following holds
	\[
		\sum_{i=1}^z k_i = 	\frac{\log n}{\log \lp( \frac{\opt}{\opt-1-\e} \rp)} < \log n \frac{\opt}{1+\e}
	\]
\noindent
	where the first sum is exactly the cost of covering all the elements but
	one\footnote{By adding $1$ to $\sum_{i=1}^z k_i$, the inequalities still
	hold.}, and for the inequality we used Lemma~\ref{lem:ineq_LB_max} with $x =
	\opt$ and $c = 1 + \e$. This result directly implies a better than $(1 -
	\e')\ln n$ approximation for Set Cover, which is impossible unless $P=NP$
	\cite{DinuSteu2014}.
		
\end{proof}

\begin{lemma}\label{lem:ineq_LB_max}
	\[ \lp( \log \lp(\frac{x}{x-c} \rp) \rp)^{-1}	< \frac{x}{c} \] 
	for any $x > c > 0$.
\end{lemma}

\begin{proof}
	First we prove the following inequality
	\begin{equation}\label{eq:log_ineq}
			\log (x+1) > \frac{x}{x+1}  \text{ for any }  x>0.
	\end{equation}
	Let $g(x) = \log(x+1) - \frac{x}{x+1}$. The derivative of $g(x)$ is $g'(x)
	= \frac{1}{x+1} - \frac{1}{(x+1)^2}$ and $g'(x)>0$ for $x>0$. Then $g$ is
	increasing and $\lim_{x\rightarrow 0} g(x)=0$, therefore $g(x)>0$ for
	$x>0$, and the inequality follows. By setting $x+1$ to be $x/(x-c)$ in
	inequality \ref{eq:log_ineq}, the lemma follows.
	
\end{proof}

\section{Proofs from Section~\ref{sec:pavspa}}  \label{sec:apndx_pavspa}

\begin{proof}[Proof of Lemma~\ref{lem:integral_bound}]
	\begin{align*}
	\int_{t-1}^{t}\lceil\beta t'\rceil dt'
		&\leq (\lceil\beta t\rceil-1)\lp(\frac{\lceil\beta t\rceil-1}{\beta}-t+1\rp)
				+\lceil\beta t\rceil\lp( t - \frac{\lceil\beta t\rceil-1}{\beta}\rp) \\
		& = t - \frac{\lceil \beta t\rceil-1}{\beta}-1+\lceil \beta t \rceil\\
		& < \beta\left(t-\frac{\lceil \beta t\rceil-1}{\beta}\right)-1+\lceil \beta t \rceil\\
		& =\beta t.
	\end{align*}

	where the first line is true since for any $t'\leq \frac{\lceil\beta
	t\rceil-1}{\beta}$, $\lceil\beta t'\rceil\leq \lceil\beta t\rceil-1$; while
	for any $t'$ such that $\frac{\lceil\beta t\rceil-1}{\beta}<t'\leq t$,
	$\lceil\beta t'\rceil= \lceil\beta t\rceil$. On the third line we used that
	$\beta>1$ and $t>\frac{\lceil\beta t\rceil-1}{\beta}$. 
	\end{proof}

\section{Proofs from Section~\ref{sec:extensions}} 
\label{sec:apndx_extension}

\begin{proof}[Proof of Lemma~\ref{lem:rank}]
	Denote by $A_j$ the span of the first $j$ elements selected. Notice that for all $i \in \boxes\setminus A_j$, 
	the probability of selecting box $i$ is $\mconst \ln k\frac{\sum_{t'\le t} z_{ist'}}{t}$. 
	Thus, the probability of selecting a box that increases the rank by $1$ at step $t\geq t_s^*$ is 
 	
	\begin{align*}
		\Pr { \text{rank } j \text{ to } j+1}
			&= 1-\prod_{i\in \boxes\setminus A_j} \left(1-\mconst\ln k \frac{\sum_{t'\leq t}z_{ist'}}{t}\right)	\\
			&\geq 1-\prod_{i\in \boxes\setminus A_j} \left(1-\mconst\ln k \frac{\sum_{t'\leq t}z_{ist'}}{2^\ell t^*_s}\right)	\\
			&\geq 1-\exp\left(-\sum_{i\in \boxes\setminus A_j}\mconst\ln k \frac{\sum_{t'\leq t}z_{ist'}}{2^\ell t^*_s}\right) \\
			& \geq1-\exp\left(-\frac{\mconst \ln k (k-j)}{2^{\ell+1} t_s^*}\right)\\
			&\geq \min\left(\frac{1}{2},\frac{\mconst\ln k (k-j)}{2^{\ell+2} t_s^*}\right),
	\end{align*}
	here the second line follows from $t\leq t_s^*$ in phase $\ell$; the third line follows from 
	$\prod_{i}(1-a_i)\leq e^{-\sum_{i}a_i}$ for any $a_1,a_2,\cdots\in[0,1]$; the fourth line follows
	from constraint (\ref{eq:LP_mat_subsets_rank}) and $y_{st}\geq\frac{1}{2}$ 
	by (\ref{eq:opt_ts_kcoverage}); the last line follows from $1-e^{-a}\geq\frac{1}{2}\min(1,a)$. Thus the expected total steps until a full rank basis is found is 
	
	\[
		\E{}{\mathcal{X}} = \sum_{j = 0}^{k-1} \E{}{\text{steps from rank } j \text{ to } j+1}
		\leq \sum_{j = 0}^{k-1} \left(\frac{2^{\ell+2}t^*_s}{\mconst(k-j)\ln k}+2\right) 
		\leq \frac{2^{\ell+2}}{\mconst} t^*_s+2k\leq \frac{2^{\ell+2}}{\mconst}t_s^*+4t_s^*.
	\]
	Here the last equality is by $t_s^*\geq\frac{k}{2}$.
\end{proof}

\bibliographystyle{alpha}
\bibliography{allrefs}

\end{document}